\documentclass[journal]{IEEEtran}
\usepackage{epsfig,graphics,subfigure,psfrag,amsmath,cases}
\usepackage{latexsym,amssymb,amsmath,empheq,algorithm}
\usepackage{graphicx}
\usepackage{epstopdf}
\usepackage{algorithmic}
\usepackage{color}
\usepackage{url}
\usepackage{scrtime}
\usepackage{amsthm}
\usepackage{datetime}
\usepackage{bm}
\usepackage{cite}
\usepackage{cleveref}
\usepackage{setspace}
\allowdisplaybreaks[4]

\title{Energy Efficiency Analysis of IRS-aided Wireless Communication Systems Under Statistical QoS Constraints: An Information-Theoretic Perspective}
\author{Wenhao Wang\thanks{W. Wang and D. Qiao are with the School of Communication and Electronic Engineering, East China Normal University, Shanghai, China (e-mail: whwang@cee.ecnu.edu.cn, dlqiao@ce.ecnu.edu.cn).}, Deli Qiao, Lei Yang, Yueying Zhan\thanks{L. Yang and Y. Zhan are with the Key Laboratory of Space Utilization, Technology and Engineering Center for Space Utilization, Chinese Academy of Sciences, Beijing, China (e-mail: yang.lei, zhanyueying@csu.ac.cn).}, and Derrick Wing Kwan Ng,~\IEEEmembership{Fellow,~IEEE}\thanks{D. W. K. Ng is with the School of Electrical Engineering and Telecommunications, University of New South Wales, Sydney, NSW 2052, Australia (e-mail: w.k.ng@unsw.edu.au).}\vspace{-0.5em}}

\newtheorem{Thm}{Theorem}
\newtheorem{Lem}{Lemma}
\newtheorem{Cor}{Corollary}

\newtheorem{T-Prob}{Transformed Problem}

\newtheorem{Remark}{Remark}

\begin{document}

\maketitle

\begin{abstract}
  This paper investigates the information-theoretic energy efficiency of intelligent reflecting surface (IRS)-aided wireless communication systems, taking into account the statistical quality-of-service (QoS) constraints on delay violation probabilities. Specifically, effective capacity is adopted to capture the maximum constant arrival rate that can be supported by a time-varying service process while fulfilling these statistical QoS requirements. We derive the minimum bit energy required for the IRS-aided wireless communication system under QoS constraints and analyze the spectral efficiency and energy efficiency tradeoff at low but nonzero signal-to-noise ratio (SNR) levels by also characterizing the wideband slope values. Our analysis demonstrates that the energy efficiency for the considered system under statistical QoS constraints can approach that for a system without QoS limitations in the low-SNR regime. Additionally, deploying a sufficiently large number of practical IRS reflecting elements can substantially reduce energy consumption required to achieve desired spectral efficiency performance in the low-power regime, even with limited bit-resolution phase shifters. Besides, we reveal that compared with the results applied to the low-power regime, higher effective capacity performance can be achieved in scenarios with sparse multipath fading while achieving the same minimum bit energy in the wideband regime.
\end{abstract}


\section{Introduction}

There are two fundamental considerations for designing future wireless communication systems: the scarcity of energy resources and the random variability of propagation channels \cite{Mingzhe}. On the one hand, wireless communication systems are expected to provide high data rates and reliable transmissions at low energy costs, leading to extensive research focused on energy-efficient system designs over the past few decades, e.g., \cite{Wesemann,Zhihan,Khowaja}. On the other hand, guaranteeing deterministic quality-of-service (QoS) proves highly challenging due to the time-varying and stochastic natures of wireless channels caused by mobility, multipath fading, and changing environments \cite{Abbou}. In general, improving energy efficiency while satisfying statistical QoS constraints is of vital importance for achieving optimal system performance and quality \cite{Changyang}.

Recently, intelligent reflecting surface (IRS) has been extensively recognized as a promising technique to improve system energy efficiency by flexibly establishing a favorable communication environment \cite{Wu2}. In particular, an IRS equipped with various reconfigurable and low-cost passive elements that can reflect desired signals by appropriately altering controllable phase shifts \cite{Gong}. By properly designing the IRS phase shift matrix, desired signals transmitted in the direct path can be constructively combined with the reflected signals to achieve substantial power gains at the targeted receivers \cite{Wu}. In fact, there have been numerous studies investigating energy-efficient designs of IRS-aided wireless communication systems, e.g., \cite{George,Chenwu,Rose}. However, these works only capture performance of classical Shannon capacity in the physical layer assuming infinite backlog, overlooking statistical QoS limitations attributed to queue length violation and buffer overflow probabilities in the data-link layer \cite{Negi}. As a result, these studies are not applicable to emerging wireless communication systems with delay-sensitive applications \cite{Amiri}, such as automatic driving, extended reality, and factory automation, etc.

It is widely recognized that effective capacity \cite{Negi} serves as an efficient metric to measure cross-layer performance in wireless communication systems, particularly under the influence of statistical QoS constraints. Specifically, effective capacity captures the maximum constant arrival rate that a time-varying service process can sustain while fulfilling a statistical QoS requirement, i.e., buffer overflow probability or delay violation probability \cite{Hsiao}. It is noteworthy that there have been a few literatures studying the energy-efficient IRS-aided wireless communication systems by exploiting effective capacity formulation, e.g., \cite{Basharat,Vincent,Moosavi}. Nonetheless, these researches merely focus on the effective rate-per-Joule as an energy efficiency metric that cannot provide the fundamental performance limits and identify the most efficient utilization of scarce energy resources \cite{Verdu}.

From an information-theoretic perspective, the energy required for reliably transmitting one bit of information, commonly referred to as \emph{bit energy}, has been extensively adopted to evaluate the energy efficiency performances of wireless communication systems \cite{Verdu,Mustafa,Gursoy}. In particular, when a system operates at low signal-to-noise ratio (SNR) levels, e.g., in the low-power or wideband regimes, minimizing bit energy is equivalent to maximizing energy efficiency \cite{Verdu,Mustafa,Gursoy}. Based on the theory, there have been numerous literatures evaluating the energy-efficient wireless communication systems under statistical QoS limitations. For example, the authors analyzed the energy efficiencies at low-SNR levels for simple-input simple-output (SISO) \cite{Mustafa} and multi-input multi-output (MIMO) \cite{Gursoy} wireless communication systems under statistical QoS constraints. Also, the tradeoff between normalized effective capacity and bit energy in the low-power or wideband regimes by characterizing the wideband slope values was investigated in \cite{Qiao}, which unveiled the impact of statistical QoS constraints on energy efficiency performance. These analytical insights facilitate the determination of the energy required to achieve specific spectral efficiency in the existence of statistical QoS limitations, which is desired for designing practical IRS-aided wireless communication systems, whereas the relevant research is still in its infancy. Additionally, it remains unclear whether IRS implementation is beneficial for significantly improving the system performance without incurring excessive energy costs, particularly due to potential spectral efficiency degradation caused by strict QoS constraints. Hence, there is a crucial need to study the tradeoff between spectral efficiency and energy efficiency for IRS-aided wireless communication systems under statistical QoS constraints in the low-power or wideband regimes.

Motivated by the above observations, we analyze the energy efficiency performance of an IRS-aided wireless communication system taking into account statistical QoS constraints from the information-theoretic perspective. The main contributions of our paper are summarized below:

\begin{itemize}
\item We derive the expression for the minimum bit energy in the low-power and wideband regimes for the IRS-aided wireless communication system, considering the existence of statistical QoS constraints. Also, we demonstrate that the energy efficiency for the considered system can approach that of a system without statistical QoS limitations at low but nonzero SNR levels.
\item To shed light on the effectiveness of IRS in optimizing energy efficiency, we further approximate the results acquired in the low-power regime under specific scenarios where a sufficiently large number of IRS reflecting elements are implemented with continuous and discrete phase shifts. Moreover, we demonstrate that deploying a sufficiently large number of the practical IRS reflecting elements can substantially reduce the energy consumption need to achieve the required spectral efficiency performance at low but nonzero SNR levels, even with limited bit-resolution phase shifters.
\item We reveal that the results acquired in the low-power regime are also applicable to scenarios with rich multipath fading in the wideband regime. Furthermore, under the assumption of employing an infinite large number of IRS reflecting elements, compared to the low-power regime, higher effective capacity performance can be achieved in scenarios characterized by sparse multipath fading where the number of the independent resolvable subchannels is bounded while satisfying the same minimum bit energy requirement in the wideband regime.
\item We prove that the expression for the minimum bit energy in scenarios with rich multipath fading is identical to that for scenarios with sparse multipath fading, where the number of subchannels grows sublinearly with the increasing bandwidth. Besides, we unveil that the energy efficiency performance of the considered system without the multipath sparsity can approach that for a system in sparse multipath fading without QoS limitations.
\end{itemize}

We organize the rest of the paper as follows. Specifically, we first introduce preliminaries including system model, effective capacity, and spectral efficiency-energy efficiency tradeoff in Section II. Then, Sections III and IV present the energy efficiency in the low-power and wideband regimes, respectively. In Sections V and VI, we present the simulation results that valid our theoretical findings and conclude the paper.

\emph{Notations:} $\Gamma \left( x \right) = \int \nolimits_{0}^{\infty} t^{x - 1} {\mathrm{exp}} ( {- t } ) {\mathrm{d}} t$ denotes the Gamma function for any $x > 0$. ${\mathbb{C}^{m \times n}}$ represents the set of an $m \times n$ complex-valued matrix, while ${\mathbb{R}}^+$ indicates a positive real-valued scalar. $\jmath = \sqrt { - 1} $ denotes an imaginary unit. $\mathbb{E}\left\{ \cdot \right\}$, $\mathrm{Pr}\left\{ \cdot \right\}$, ${\mathrm{exp}} \left(  \cdot  \right)$, $\left|  \cdot  \right|$, ${\left(  \cdot  \right)^{-1}}$, ${\left(  \cdot  \right)^T}$, ${\left(  \cdot  \right)^H}$, and ${\left(  \cdot  \right)^ * }$ denote statistical expectation, occurrence probability of an event, exponential, absolute value, inverse, transpose, conjugate transpose, and conjugate of a input, respectively. $y = o ( x )$ means $\lim \nolimits_{y \rightarrow 0, x \rightarrow 0} y / x = 0$. ${\mathrm{diag}}\left( {\mathbf{x}} \right)$ represents a diagonal matrix whose main diagonal elements are the associated elements in vector ${\mathbf{x}}$. $ \triangleq $ and $ \sim $ stand for ``defined as'' and ``distributed as'', respectively. $\mathrm{Nakagami} ( \alpha,\beta )$ and $\mathrm{Gamma} ( \alpha,\beta )$ denote the Nakagami and Gamma distributions, respectively, where $\alpha$ and $\beta$ are the corresponding shape and scaling parameters.

\section{Preliminaries}

\subsection{System Model}

We consider an IRS-aided wireless communication system composed of one transmitter, one receiver, and an IRS. We assume that both the transmitter and the receiver are single-antenna devices, while the IRS is equipped with $N$ reflecting elements. Denote ${\mathcal{N}}$ as the set of reflecting elements and ${\mathbf{\Theta }} = {\mathrm{diag}}\big( {{\mathrm{exp}} \big({\jmath{\theta _1}} \big), \ldots, {\mathrm{exp}} \big({\jmath{\theta _n}} \big), \ldots, {\mathrm{exp}} \big({\jmath{\theta _N}} \big)} \big) \in {\mathbb{C}^{N \times N}}$ as the reflection-coefficients matrix of the IRS, where $\theta_n \in \left[ {0,2\pi } \right)$ is the phase shift of the $n$-th reflecting element, $\forall n \in {\mathcal{N}}$. We consider that all channels of the considered system follow as quasi-static flat block fading. Also, perfect channel state information (CSI) is available at the transmitter and the receiver \cite{Zheng,Cui,Beixiong}. The received signal at the receiver is given by
\setlength\abovedisplayskip{4pt}
\setlength\belowdisplayskip{4pt}
\begin{align}
  y = \big( \sqrt{\ell_{\mathrm{f}} \ell_{\mathrm{g}}} {\mathbf{f}}^T {\mathbf{\Theta }} {\mathbf{g}} + \sqrt{\ell_{\mathrm{h}}} h \big)  x + z,
\end{align}
where ${\mathbf{g}} \in {\mathbb{C}^{N \times 1}}$, ${\mathbf{f}} \in {\mathbb{C}^{N \times 1}}$, and $h \in {\mathbb{C}}$ are small-scale fading coefficients of the transmitter-IRS, IRS-receiver, and direct transmitter-receiver links, respectively, while $\ell_{\mathrm{g}} \in {\mathbb{R}}^+$, $\ell_{\mathrm{f}} \in {\mathbb{R}}^+$, and $\ell_{\mathrm{h}} \in {\mathbb{R}}^+$ denote the corresponding distance-dependent path-losses, respectively. Specifically, we assume that all the involved links experience Nakagami-$m$ fading with the corresponding shape parameters $m_{\mathrm{g}}$, $m_{\mathrm{f}}$, and $m_h$, i.e., $\big| g_n \big| \sim \mathrm{Nakagami} \big( m_{\mathrm{g}},1 \big)$, $\big| f_n \big| \sim \mathrm{Nakagami} \big( m_{\mathrm{f}},1 \big)$, and $\big| h \big| \sim \mathrm{Nakagami} \big( m_{\mathrm{h}},1 \big)$, respectively, where $f_n$ and $g_n$ are the $n$-th element of small-scale fading coefficient vectors ${\mathbf{f}}$ and ${\mathbf{g}}$, respectively, $\forall n \in {\mathcal{N}}$. Furthermore, $x \in {\mathbb{C}}$ denotes the transmitted signal with the following average energy constraint, i.e., $\mathbb{E}\big\{ | x |^2 \big\} \leqslant P / B$, where $B$ and $P$ denote the system bandwidth and the upper bound for average transmit power budget, respectively. Besides, $z$ indicates the additive white Gaussian noise (AWGN) at the receiver with $\mathbb{E}\big\{ | z |^2 \big\} = N_0$, where $N_0$ denotes the corresponding noise power. As such, the average transmit SNR for the considered system is expressed as ${\mathrm{SNR}} = P/(N_0 B)$.

\subsection{Effective Capacity}

We introduce a crucial concept of effective capacity \cite{Negi} that provides the maximum constant arrival rate supported by a given time-varying service process while fulfilling a statistical QoS requirement specified by a QoS exponent $\mu$, which is defined as \cite{Negi}
\begin{align}
  \label{mu}
  \lim\limits_{Q_{\mathrm{max}} \rightarrow \infty} \frac{{\mathrm{ln}} \Big( {\mathrm{Pr}} \big( Q \geqslant Q_{\mathrm{max}} \big) \Big)}{Q_{\mathrm{max}}} = - \mu,
\end{align}
where $Q$ and $Q_{\mathrm{max}}$ denote the stationary queue length and the upper bound of queue length, respectively. Based on the above definition, we have the following approximation for the queue-overflow probability, i.e., ${\mathrm{Pr}} \big( Q \geqslant Q_{\mathrm{max}} \big) \approx {\mathrm{exp}} \big( {- \mu Q_{\mathrm{max}}} \big)$. Similarly, by denoting the steady-state delay in the buffer and the upper bound of communication delay as $D$ and $D_{\mathrm{max}}$, respectively, the delay violation probability can be approximated as ${\mathrm{Pr}} \big( D \geqslant D_{\mathrm{max}} \big) \approx {\mathrm{exp}} \big( {- \mu \delta D_{\mathrm{max}}} \big)$, where $\delta$ is related to the arrival and service processes \cite{Senem}. Note that $\mu \rightarrow \infty$ indicates extremely stringent QoS constraints while $\mu \rightarrow 0$ relaxes QoS limitations. As such, effective capacity is regarded as the maximum throughput under the queue length constraint, i.e., ${\mathrm{Pr}} \big( Q \geqslant Q_{\mathrm{max}} \big) \leqslant {\mathrm{exp}} \big( {- \mu Q_{\mathrm{max}}} \big)$ for $Q_{\mathrm{max}}$, or the delay constraint, i.e., ${\mathrm{Pr}} \big( D \geqslant D_{\mathrm{max}} \big) \leqslant {\mathrm{exp}} \big( {- \mu \delta D_{\mathrm{max}}} \big)$ for $D_{\mathrm{max}}$ \cite{Cenk}. Specifically, the effective capacity is denoted as \cite{Velipasalar}
\begin{align}
  C_{\mathrm{E}} = - \lim\limits_{t \rightarrow \infty} \frac{1}{\mu t} {\mathrm{ln}} \Big( \mathbb{E} \Big\{ {\mathrm{exp}} \big( {- \mu S \big[ t \big]} \big) \Big\} \Big),
\end{align}
where $S \big[ t \big] = \sum\nolimits_{i=1}^t r\big[ i \big]$ indicates the time-accumulated process, where $\big\{r\big[ i \big], i = 1, 2, \ldots \big\}$ represents the discrete-time stationary
and ergodic stochastic service process. Based on the block fading assumption, the corresponding channel realizations remain constant over $T$ durations in one block and vary independently from one block to another. As such, the effective capacity (bit/s) can be simplified to
\begin{align}
  C_{\mathrm{E}} = - \frac{1}{\mu T} {\mathrm{ln}} \Big( \mathbb{E} \Big\{ {\mathrm{exp}} \big( {- \mu T R} \big) \Big\} \Big),
\end{align}
where $R$ denotes the instantaneous service rate. For a given QoS exponent $\mu$, the effective capacity (bit/s/Hz) normalized by the system bandwidth, i.e., $B$, for the considered system is expressed as \eqref{C_E} at the top of next page,
\begin{figure*}[ht]
\begin{align}
  \label{C_E}
  C_{\mathrm{E}} \big( {\mathrm{SNR}}, \mu, {\mathbf{\Theta }} \big) = - \frac{1}{\mu T B} {\mathrm{ln}} \Big( \mathbb{E} \Big\{ {\mathrm{exp}} \Big( {- \mu T B \mathop {\mathrm{max}}\limits_{{\mathbf{\Theta}}, |\theta_m|=1} {\mathrm{log}}_2 \Big( 1 + {\mathrm{SNR}} \big| \bar h \big|^2 \Big)}  \Big) \Big\} \Big).
\end{align}\hrulefill\vspace*{-4mm}
\end{figure*}
where $\bar h \in {\mathbb{C}}$ is defined as $\bar h = \sqrt{\ell_{\mathrm{f}} \ell_{\mathrm{g}}} {\mathbf{f}}^T {\mathbf{\Theta }} {\mathbf{g}} + \sqrt{\ell_{\mathrm{h}}} h$. Here, we acquire the optimal solution to the above problem in \eqref{C_E} via adopting the following lemma.
\begin{Lem}
  \label{optimal_theta}
  For a given transmit SNR, the $n$-th optimal phase shift adopted at the IRS is expressed as $\theta_n^{\star} = \arg\big( h \big) - \arg\big( f_n \big) - \arg\big( g_n \big)$, where $f_n$ and $g_n$ are the $n$-th element of channel vectors ${\mathbf{f}}$ and ${\mathbf{g}}$, respectively, $\forall n \in {\mathcal{N}}$.
\end{Lem}
\begin{proof}
  Please refer to Appendix A for details.
\end{proof}

By substituting the optimal solution obtained in \textbf{Lemma \ref{optimal_theta}} into $\bar h$ in \eqref{C_E}, we have $\big| \bar h \big|^2 = \big( \sqrt{\ell_{\mathrm{h}}} \xi_{\mathrm{d}} + \sqrt{\ell_{\mathrm{f}} \ell_{\mathrm{g}}} \xi_{\mathrm{r}} \big)^2 \triangleq \xi$, where $\xi_{\mathrm{d}} = \big| h \big|$ and $\xi_{\mathrm{r}} = \sum\nolimits_{n\in {\mathcal{N}}} \big| f_n \big| \big| g_n \big|$. Hence, the normalized effective capacity (bit/s/Hz) in \eqref{C_E} can be rewritten as
\begin{align}
  \label{normalized_C_E}
  &C_{\mathrm{E}} \big( {\mathrm{SNR}}, \mu \big) \notag\\
  & = - \frac{1}{\mu T B} {\mathrm{ln}} \Big( \mathbb{E} \Big\{ {\mathrm{exp}} \Big( {- \mu T B {\mathrm{log}}_2 \big( 1 + {\mathrm{SNR}} \xi \big)} \Big) \Big\} \Big).
\end{align}

\subsection{Spectral Efficiency and Energy Efficiency Tradeoff}

In this paper, we evaluate the energy efficiency performance for the considered system by exploring bit energy, which is defined as energy-per-bit normalized to background noise spectral level \cite{Verdu} denoted as ${E_{\mathrm{b}}}/{N_0}$ (dB), where $E_{\mathrm{b}}$ is the transmit energy-per-bit. In particular, the minimum bit energy has been widely recognized as a performance metric for reliable communications in low-SNR regime \cite{Verdu}. Also, wideband slope, which is defined as the curve slope of spectral efficiency versus bit energy at the minimum bit energy, has emerged as a classical concept that facilitates analyzing the system energy efficiency at low-power and wideband levels. Different from the previous works focusing on the tradeoff between spectral efficiency and energy efficiency based on the Shannon capacity, e.g., \cite{Verdu}, \cite{Tulino}, \cite{Poor}, we aim to perform an analysis of energy efficiency in IRS-aided wireless communication systems under statistical QoS limitations. Specifically, we analyze the normalized effective capacity-bit energy tradeoff. Moreover, the effective capacity provides a characterization of the arrival process, which is also regarded as a measure of the average transmission rate, since the average arrival rate is identical to the average departure rate when the queue is in steady-state \cite{Negi,Amiri,Hsiao}.


In the following, we investigate the minimum bit energy and the tradeoff between the normalized effective capacity and bit energy to investigate the energy efficiency of the considered system. Specifically, for a given $\mu$, the minimum bit energy under statistical QoS constraints is expressed as \cite{Mustafa}
\begin{align}
  \label{min_bit_energy}
  {\frac{E_{\mathrm{b}}}{N_0}}_{\min} = \lim\limits_{{\mathrm{SNR}} \rightarrow 0} \frac{\mathrm{SNR}}{C_{\mathrm{E}} \big( {\mathrm{SNR}} \big) } = \frac{1}{\dot{C}_{\mathrm{E}} \big( 0 \big)},
\end{align}
where $\dot{C}_{\mathrm{E}} \big( 0 \big)$ denotes the first derivative of the normalized effective capacity at ${\mathrm{SNR}} = 0$. Additionally, at the minimum bit energy in \eqref{min_bit_energy}, the curve slope of the normalized effective capacity versus bit energy in bit/s/Hz/($3$ dB) is given by \cite{Gursoy}
\begin{align}
  \label{wideband_slope}
  S_0 &= \lim\limits_{\frac{E_{\mathrm{b}}}{N_0} \downarrow {\frac{E_{\mathrm{b}}}{N_0}}_{\min}} \frac{C_{\mathrm{E}} \Big( \displaystyle{\frac{E_{\mathrm{b}}}{N_0}} \Big) }{10 \log_{10} \displaystyle{\frac{E_{\mathrm{b}}}{N_0}} - 10  \log_{10} \displaystyle{{\frac{E_{\mathrm{b}}}{N_0}}_{\min}} } 10  \log_{10} 2 \notag\\
  & = - \frac{2 \big( \dot{C}_{\mathrm{E}} \big( 0 \big) \big)^2}{\ddot{C}_{\mathrm{E}} \big( 0 \big)} {\mathrm{ln}} 2,
\end{align}
where $\ddot{C}_{\mathrm{E}} \big( 0 \big)$ denotes the second derivative of the normalized effective capacity at ${\mathrm{SNR}} = 0$. Note that the minimum bit energy and the wideband slope for the considered system offer a linear approximation for the curve of the normalized effective capacity versus bit energy at low SNR, which indicates the tradeoff between spectral efficiency and energy efficiency, i.e.,
\begin{align}
  C_{\mathrm{E}} \Big( {\frac{E_{\mathrm{b}}}{N_0}} \Big) &= \frac{S_0}{10 \log_{10} 2} \bigg( 10 \log_{10} \frac{E_{\mathrm{b}}}{N_0} - 10 \log_{10} {\frac{E_{\mathrm{b}}}{N_0}}_{\min} \bigg) \notag\\
  & \ \ \ \ + o \bigg( {\frac{E_{\mathrm{b}}}{N_0}} - {\frac{E_{\mathrm{b}}}{N_0}}_{\min} \bigg).
\end{align}
To further analyze the energy efficiency for the considered system in the low-SNR regime, i.e., ${\mathrm{SNR}} = P/(N_0 B) \rightarrow 0$, we will consider the low-power regime for a given bandwidth $B$ and the wideband regime for a given average transmit power $P$ in the following sections, respectively.

\section{Energy Efficiency in the Low-Power Regime}

In this section, we aim to study the tradeoff between spectral efficiency and energy efficiency by exploring the minimum bit energy and wideband slope as the average transmit power $P$ diminishes while system bandwidth $B$ is fixed. In particular, for a given QoS exponent $\mu$, we first express the normalized effective capacity (bit/s/Hz) for the considered system in \eqref{normalized_C_E} as a function of ${\mathrm{SNR}}$. Then, we can obtain the minimum bit energy and the wideband slope for the considered system by the following theorem.
\begin{Thm}
  \label{low_power}
  The minimum bit energy and the wideband slope for the considered system in the low-power regime, respectively, are expressed as
  \begin{align}
  \label{minimum_bit_energy_lowpower}
  \bar {\frac{E_{\mathrm{b}}}{N_0}}_{\min} &= \frac{{\mathrm{ln}}2}{\mathbb{E} \big\{ \xi \big\} } {\text{ and}} \\
  \label{wideband_slope_lowpower}
  \bar S_0 &= \frac{2 {\mathrm{ln}}2 \big( \mathbb{E} \big\{ \xi \big\} \big)^2 }{\big( \mu T B + {\mathrm{ln}}2 \big) \mathbb{E} \big\{ \xi^2 \big\} - \mu T B \big( \mathbb{E} \big\{ \xi \big\} \big)^2}.
  \end{align}
\end{Thm}
\begin{proof}
  Please refer to Appendix B for details.
\end{proof}

\begin{Remark}
  \label{low_power_no_QoS}
  As can be observed, the minimum bit energy for the considered system in the low-power regime in \eqref{minimum_bit_energy_lowpower} is independent of the QoS exponent $\mu$. Also, we observe that by assuming $\mu \rightarrow 0$, the wideband slope for the considered system in the low-power regime in \eqref{wideband_slope_lowpower} can be further expressed as
  \begin{align}
  \bar S_0 = \frac{2 \big( \mathbb{E} \big\{ \xi \big\} \big)^2 }{\mathbb{E} \big\{ \xi^2 \big\}}.
  \end{align}
  As such, we conclude that the energy efficiency for the considered system can approach that for a system without QoS limitations in the low-power regime. Besides, the system may incur spectral efficiency degradation or more additional energy cost for given system performance, due to the existence of QoS constraints.
\end{Remark}

Next, we consider a special scenario for the considered system with a sufficiently large number of the IRS reflecting elements in the low-power regime. Specifically, we can acquire the minimum bit energy and the wideband slope in this scenario by the following corollary.
\begin{Cor}
  \label{large_number_elements}
  By assuming that the number of the IRS reflecting elements, i.e., $N$, is sufficiently large, the minimum bit energy in \eqref{minimum_bit_energy_lowpower} and the wideband slope in \eqref{wideband_slope_lowpower} for the considered system in the low-power regime are further approximated by
  \begin{align}
  \label{minimum_bit_energy_large_number_elements}
  \hat {\frac{E_{\mathrm{b}}}{N_0}}_{\min} &= \frac{{\mathrm{ln}}2}{\ell_{\mathrm{h}} {\mathbb{E}} \big\{ \xi_{\mathrm{d}}^2 \big\} + \ell_{\mathrm{f}} \ell_{\mathrm{g}} {\mathbb{E}} \big\{ \xi_{\mathrm{r}}^2 \big\} + 2 \sqrt{\ell_{\mathrm{h}} \ell_{\mathrm{f}} \ell_{\mathrm{g}}} {\mathbb{E}} \big\{ \xi_{\mathrm{d}} \big\} {\mathbb{E}} \big\{ \xi_{\mathrm{r}} \big\}}
  \end{align}
  and \eqref{wideband_slope_large_number_elements} at the top of next page, respectively, where ${\mathbb{E}} \big\{ \xi_{\mathrm{d}}^k \big\}$ and ${\mathbb{E}} \big\{ \xi_{\mathrm{r}}^k \big\}$ denote the $k$-th moments of $\xi_{\mathrm{d}}$ and $\xi_{\mathrm{r}}$, respectively.
  \begin{figure*}[ht]
  \begin{align}
  \label{wideband_slope_large_number_elements}
  &\hat S_0 = 2 {\mathrm{ln}}2 \Big( {\ell_{\mathrm{h}} {\mathbb{E}} \big\{ \xi_{\mathrm{d}}^2 \big\} + \ell_{\mathrm{f}} \ell_{\mathrm{g}} {\mathbb{E}} \big\{ \xi_{\mathrm{r}}^2 \big\} + 2 \sqrt{\ell_{\mathrm{h}} \ell_{\mathrm{f}} \ell_{\mathrm{g}}} {\mathbb{E}} \big\{ \xi_{\mathrm{d}} \big\} {\mathbb{E}} \big\{ \xi_{\mathrm{r}} \big\}} \Big)^2 \bigg( \big( \mu T B + {\mathrm{ln}}2 \big) \Big( \ell_{\mathrm{h}}^2 {\mathbb{E}} \big\{ \xi_{\mathrm{d}}^4 \big\} + \ell_{\mathrm{f}}^2 \ell_{\mathrm{g}}^2 {\mathbb{E}} \big\{ \xi_{\mathrm{r}}^4 \big\} + 6 \ell_{\mathrm{h}} \ell_{\mathrm{f}} \ell_{\mathrm{g}} {\mathbb{E}} \big\{ \xi_{\mathrm{d}}^2 \big\} {\mathbb{E}} \big\{ \xi_{\mathrm{r}}^2 \big\} \notag\\
  & + 4 \ell_{\mathrm{h}} \sqrt{\ell_{\mathrm{h}} \ell_{\mathrm{f}} \ell_{\mathrm{g}}} {\mathbb{E}} \big\{ \xi_{\mathrm{d}}^3 \big\} {\mathbb{E}} \big\{ \xi_{\mathrm{r}} \big\} + 4 \ell_{\mathrm{f}} \ell_{\mathrm{g}} \sqrt{\ell_{\mathrm{h}} \ell_{\mathrm{f}} \ell_{\mathrm{g}}} {\mathbb{E}} \big\{ \xi_{\mathrm{d}} \big\} {\mathbb{E}} \big\{ \xi_{\mathrm{r}}^3 \big\} \Big) - \mu T B \Big( \ell_{\mathrm{h}} {\mathbb{E}} \big\{ \xi_{\mathrm{d}}^2 \big\} + \ell_{\mathrm{f}} \ell_{\mathrm{g}} {\mathbb{E}} \big\{ \xi_{\mathrm{r}}^2 \big\} + 2 \sqrt{\ell_{\mathrm{h}} \ell_{\mathrm{f}} \ell_{\mathrm{g}}} {\mathbb{E}} \big\{ \xi_{\mathrm{d}} \big\} {\mathbb{E}} \big\{ \xi_{\mathrm{r}} \big\} \Big)^2 \bigg)^{-1}.
  \end{align}\hrulefill\vspace*{-4mm}
  \end{figure*}
\end{Cor}
\begin{proof}
  Please refer to Appendix C for details.
\end{proof}

\begin{Remark}
  Note that the minimum bit energy for conventional wireless communication systems without deploying IRS in the low-power regime \cite{Mustafa}, i.e., $\hat {\frac{E_{\mathrm{b}}}{N_0}}_{\min, 0}$, is much larger than the counterpart for the considered system in \eqref{minimum_bit_energy_large_number_elements}, especially for the scenario with a sufficiently large number of the IRS reflecting elements, i.e.,
  \begin{align}
  \hat {\frac{E_{\mathrm{b}}}{N_0}}_{\min, 0} = \frac{{\mathrm{ln}}2}{\ell_{\mathrm{h}} {\mathbb{E}} \big\{ \xi_{\mathrm{d}}^2 \big\}} \geqslant \hat {\frac{E_{\mathrm{b}}}{N_0}}_{\min},
  \end{align}
  which indicates the superiority of improving the system energy efficiency by introducing IRS, despite the fact that there is the QoS limitations.
\end{Remark}

\begin{Remark}
  \label{remark_infiniteN_lower_power}
  By considering a special case with an infinite large number of the IRS reflecting elements, i.e., $N \rightarrow \infty$, we can further approximate the minimum bit energy in \eqref{minimum_bit_energy_large_number_elements} and the wideband slope in \eqref{wideband_slope_large_number_elements} for the considered system in the low-power regime as
  \begin{align}
  \hat {\frac{E_{\mathrm{b}}}{N_0}}_{\min} = 0 {\text{ and }} \hat S_0 = 2,
  \end{align}
  respectively. As such, we conclude that deploying IRS is able to compensate the spectral efficiency degradation to a certain extent, caused by the stringent QoS limitations. Also, deploying a sufficiently large number of IRS reflecting elements can achieve high spectral efficiency at low bit energy levels under statistical QoS constraints.
\end{Remark}

Based on the above special scenario for the considered system, we further consider a practical implementation that the phase shift adopted at the $n$-th IRS reflecting element can merely admit one of the $2^b$ discrete value, where $b$ is a constant bit resolution for uniformly quantizing the phase shift interval, i.e., $\theta_n \in \big[ 0,2 \pi \big)$, $\forall n \in {\mathcal{N}}$. Therefore, the set of the discrete phase shift of each reflecting element is expressed as ${\mathcal{F}} \triangleq \big\{ 0, \Delta \theta, \ldots, \left( L-1 \right) \Delta \theta \big\}$, where $\Delta \theta = 2 \pi / 2^b$. Here, we define the discrete phase shift of the $n$-th IRS reflecting element and the corresponding quantization error as $\bar \theta_n$ and $\hat \theta_n = \bar \theta_n - \theta_n^{\star}$, respectively, where $\theta_n^{\star}$ is obtained in \textbf{Lemma \ref{optimal_theta}}. As such, we rewrite the channel coefficient gain as $\bar \xi = \big( \sqrt{\ell_{\mathrm{h}}} \xi_{\mathrm{d}} + \sqrt{\ell_{\mathrm{f}} \ell_{\mathrm{g}}} \bar \xi_{\mathrm{r}} \big)^2$, where $\bar \xi_{\mathrm{r}} = \sum\nolimits_{n\in {\mathcal{N}}} \big| f_n \big| \big| g_n \big| {\mathrm{exp}} \big({\jmath{\hat \theta_n}} \big)$. Note that $\hat \theta_n$ is independently and uniformly distributed in $\big[ {- \pi/2^b,\pi/2^b} \big)$ \cite{Qingqing2}. As such, the minimum bit energy and the wideband slope under the assumption of the discrete phase shifts can be acquired by the following theorem.
\begin{Thm}
  \label{low_power_discrete}
  The minimum bit energy and the wideband slope for the considered system with discrete phase shifters in the low-power regime, respectively, are expressed as
  \begin{align}
  \label{minimum_bit_energy_lowpower_discrete}
  \breve {\frac{E_{\mathrm{b}}}{N_0}}_{\min} &= \frac{{\mathrm{ln}}2}{\mathbb{E} \big\{ \bar \xi \big\} } {\text{ and}} \\
  \label{wideband_slope_lowpower_discrete}
  \breve S_0 &= \frac{2 {\mathrm{ln}}2 \big( \mathbb{E} \big\{ \bar \xi \big\} \big)^2 }{\big( \mu T B + {\mathrm{ln}}2 \big) \mathbb{E} \big\{ \bar \xi^2 \big\} - \mu T B \big( \mathbb{E} \big\{ \bar \xi \big\} \big)^2}.
  \end{align}
\end{Thm}
\begin{proof}
  The proof of the theorem similarly follows Appendix B. Here, we omit the details of the proof due to the space limitation of this paper.
\end{proof}

\begin{Remark}
  Note that the minimum bit energy for the considered system with discrete phase shifters in the low-power regime in \eqref{minimum_bit_energy_lowpower_discrete} is much larger than the counterpart with continuous phase shifters in \eqref{minimum_bit_energy_lowpower}, especially for the limited bit-resolution implementation, i.e.,
  \begin{align}
  \breve {\frac{E_{\mathrm{b}}}{N_0}}_{\min} &= \frac{{\mathrm{ln}}2}{\ell_{\mathrm{h}} {\mathbb{E}} \big\{ \xi_{\mathrm{d}}^2 \big\} + \ell_{\mathrm{f}} \ell_{\mathrm{g}} {\mathbb{E}} \big\{ \bar \xi_{\mathrm{r}}^2 \big\} + 2 \sqrt{\ell_{\mathrm{h}} \ell_{\mathrm{f}} \ell_{\mathrm{g}}} {\mathbb{E}} \big\{ \xi_{\mathrm{d}} \bar \xi_{\mathrm{r}} \big\}} \notag\\
  &\overset{(a)} \geqslant \frac{{\mathrm{ln}}2}{\ell_{\mathrm{h}} {\mathbb{E}} \big\{ \xi_{\mathrm{d}}^2 \big\} + \ell_{\mathrm{f}} \ell_{\mathrm{g}} {\mathbb{E}} \big\{ \xi_{\mathrm{r}}^2 \big\} + 2 \sqrt{\ell_{\mathrm{h}} \ell_{\mathrm{f}} \ell_{\mathrm{g}}} {\mathbb{E}} \big\{ \xi_{\mathrm{d}} \xi_{\mathrm{r}} \big\}}  \notag\\
  & = \frac{{\mathrm{ln}}2}{\mathbb{E} \big\{ \xi \big\} } = \bar {\frac{E_{\mathrm{b}}}{N_0}}_{\min},
  \end{align}
  where ($a$) is acquired by exploiting the following inequalities, i.e., \eqref{bar_xi_r_xi_r} and \eqref{xi_d_bar_xi_r_xi_d_xi_r} at the top of next page, where ($b$) and ($d$) can be obtained by the fact that $|h|$, $|f_n|$, $|g_n|$, and ${\mathrm{exp}} \big({\jmath{\hat \theta_n}} \big)$ are independent with each other as well as $\mathbb{E} \big\{ {\mathrm{exp}} \big({\jmath{\hat \theta_n}} \big) \big\} = \mathbb{E} \big\{ {\mathrm{exp}} \big({- \jmath{\hat \theta_n}} \big) \big\} = 2^b/ \pi \sin \big( \pi / 2^b \big)$, while ($c$) and ($e$) are explained by the fact that $2^b/ \pi \sin \big( \pi / 2^b \big) \leqslant 1$ for any constant $b > 0$. As can be observed, the equality in (a) holds if and only if $b \rightarrow \infty$ that represents the continuous phase shifts without quantization errors as $\lim \nolimits_{b \rightarrow \infty} 2^b/ \pi \sin \big( \pi / 2^b \big) = 1$. As such, we conclude that deploying IRS with discrete phase shifters for the considered system may incur more stringent energy requirements for reliable communications compared with implementing that with the continuous counterpart, despite the fact that there is the QoS limitations.
  \begin{figure*}[ht]
  \begin{align}
  \label{bar_xi_r_xi_r}
  {\mathbb{E}} \Big\{ \bar \xi_{\mathrm{r}}^2 \Big\} &= {\mathbb{E}} \bigg\{ \Big| \sum\limits_{n\in {\mathcal{N}}} \big| f_n \big| \big| g_n \big| {\mathrm{exp}} \big({\jmath{\hat \theta_n}} \big) \Big|^2 \bigg\} = {\mathbb{E}} \bigg\{ \sum\limits_{n\in {\mathcal{N}}} \sum\limits_{i \in {\mathcal{N}\backslash }\left\{ n \right\}} \big| f_n \big| \big| g_n \big| \big| f_i \big| \big| g_i \big| {\mathrm{exp}} \big({\jmath{\hat \theta_n} - \jmath{\hat \theta_i}} \big) \bigg\} + {\mathbb{E}} \bigg\{ \sum\limits_{n\in {\mathcal{N}}} \big| f_n \big|^2 \big| g_n \big|^2 \bigg\} \notag\\
  &\overset{(b)} = \Big( \frac{2^b}{\pi} \sin \Big( \frac{\pi}{2^b} \Big) \Big)^2 {\mathbb{E}} \bigg\{ \sum\limits_{n\in {\mathcal{N}}} \sum\limits_{i \in {\mathcal{N}\backslash }\left\{ n \right\}} \big| f_n \big| \big| g_n \big| \big| f_i \big| \big| g_i \big| \bigg\} + {\mathbb{E}} \bigg\{ \sum\limits_{n\in {\mathcal{N}}} \big| f_n \big|^2 \big| g_n \big|^2 \bigg\} \notag\\
  &\overset{(c)} \leqslant {\mathbb{E}} \bigg\{ \sum\limits_{n\in {\mathcal{N}}} \sum\limits_{i \in {\mathcal{N}\backslash }\left\{ n \right\}} \big| f_n \big| \big| g_n \big| \big| f_i \big| \big| g_i \big| \bigg\} + {\mathbb{E}} \bigg\{ \sum\limits_{n\in {\mathcal{N}}} \big| f_n \big|^2 \big| g_n \big|^2 \bigg\} = {\mathbb{E}} \bigg\{ \Big( \sum\limits_{n\in {\mathcal{N}}} \big| f_n \big| \big| g_n \big| \Big)^2 \bigg\} = {\mathbb{E}} \Big\{ \xi_{\mathrm{r}}^2 \Big\}. \\
  \label{xi_d_bar_xi_r_xi_d_xi_r}
  {\mathbb{E}} \Big\{ \xi_{\mathrm{d}} \bar \xi_{\mathrm{r}} \Big\} &= {\mathbb{E}} \bigg\{ \big| h \big| \sum\limits_{n\in {\mathcal{N}}} \big| f_n \big| \big| g_n \big| {\mathrm{exp}} \big({\jmath{\hat \theta_n}} \big) \bigg\} \overset{(d)} = \Big( \frac{2^b}{\pi} \sin \Big( \frac{\pi}{2^b} \Big) \Big) {\mathbb{E}} \bigg\{ \big| h \big| \sum\limits_{n\in {\mathcal{N}}} \big| f_n \big| \big| g_n \big| \bigg\} \overset{(e)} \leqslant {\mathbb{E}} \bigg\{ \big| h \big| \sum\limits_{n\in {\mathcal{N}}} \big| f_n \big| \big| g_n \big| \bigg\} = {\mathbb{E}} \Big\{ \xi_{\mathrm{d}} \xi_{\mathrm{r}} \Big\}.
  \end{align}\hrulefill\vspace*{-4mm}
  \end{figure*}

\end{Remark}

\begin{Cor}
  \label{discrete_phase_shifts}
  By assuming that the number of the IRS reflecting elements, i.e., $N$, is sufficiently large while the quantization error of the discrete phase shift of the $n$-th IRS reflecting element, i.e., $\hat \theta_n$, is independently and uniformly distributed in $\big[ {- \pi/2^b,\pi/2^b} \big)$, the minimum bit energy in \eqref{minimum_bit_energy_lowpower_discrete} and the wideband slope in \eqref{wideband_slope_lowpower_discrete} for the considered system in the low-power regime are further approximated by $\check {\frac{E_{\mathrm{b}}}{N_0}}_{\min}$ and $\check S_0$, respectively, which share the similar structures with $\hat {\frac{E_{\mathrm{b}}}{N_0}}_{\min}$ in \eqref{minimum_bit_energy_large_number_elements} and $\hat S_0$ in \eqref{wideband_slope_large_number_elements}.
\end{Cor}
\begin{proof}
  Please refer to Appendix D for details.
\end{proof}

\begin{Remark}
  By considering a special case with the infinite large number of the IRS reflecting elements, i.e., $N \rightarrow \infty$, we obtain the same results for the considered system with implementing the discrete phase shifters in the low-power regime as shown in \textbf{Remark \ref{remark_infiniteN_lower_power}}. As can be observed, deploying the sufficiently large number of the practical IRS reflecting elements can significantly reduce the energy consumption for the required spectral efficiency performance at low bit energy levels, even with limited bit-resolution phase shifters.
\end{Remark}

\section{Energy Efficiency in the Wideband Regime}

In this section, we aim to investigate the energy efficiency as system bandwidth $B$ increases while the average transmit power $P$ is fixed. In fact, the assumption of the flat fading will no longer be ensured as the system bandwidth grows. To obstacle this circumvent, we perform an energy efficiency analysis based on the parallel, non-interacting, and narrowband subchannels divided by the wideband channel, each of which experiences the independent flat fading. We consider two scenarios for analyzing energy efficiency in the wideband regime in the following.

\emph{1) Rich multipath fading:} In the scenario with rich multipath fading, as system bandwidth $B$ increases, the number of the independent resolvable subchannels, i.e., $N_{\mathrm{c}}$, grows \emph{linearly} while the bandwidth of each subchannel, i.e., coherence bandwidth $B_{\mathrm{c}}$, remains constant, where $B = N_{\mathrm{c}} B_{\mathrm{c}}$. In other words, the power allocation for each subchannel can gradually reduce to zero with increasing bandwidth. Here, for a given $\mu$, we rewrite the normalized effective capacity (bit/s/Hz) for the considered system in \eqref{normalized_C_E} with respect to ${\mathrm{SNR}}$ as
\begin{align}
  \label{normalized_C_E_rich_multipath}
  C_{\mathrm{E}} \big( {\mathrm{SNR}} \big) = - \frac{1}{\mu T B_{\mathrm{c}}} {\mathrm{ln}} \Big( \mathbb{E} \Big\{ {\mathrm{exp}} \big( {- \mu T B_{\mathrm{c}} {\mathrm{log}}_2 \big( 1 + {\mathrm{SNR}} \xi \big) } \big) \Big\} \Big).
\end{align}
We observe that the expressions of the normalized effective capacity by fixing $B$ in \eqref{normalized_C_E} and that by fixing $B_{\mathrm{c}}$ in \eqref{normalized_C_E_rich_multipath} share a similar structure. As such, we conclude that the results acquired in the low-power regime is also applicable to the scenario with rich multipath fading in the wideband regime.

\emph{2) Sparse multipath fading:} Different from the scenario with rich multipath fading, as system bandwidth $B$ increases, the number of the independent resolvable subchannels, i.e., $N_{\mathrm{c}}$, grows at most \emph{sublinearly} while the bandwidth of each subchannel, i.e., coherence bandwidth $B_{\mathrm{c}}$, increases in the scenario with sparse multipath fading \cite{Raghavan}. Here, we consider two following special cases for analyzing energy efficiency in the scenario with sparse multipath fading in the wideband regime. In particular, as for \textbf{case I} where the number of the independent resolvable subchannels is bounded, $B_{\mathrm{c}}$ grows linearly while $N_{\mathrm{c}}$ remains constant with increasing bandwidth $B$ \cite{Telatar}. As for \textbf{case II} where the number of the independent resolvable subchannels increases but only sublinearly, both $B_{\mathrm{c}}$ and $N_{\mathrm{c}}$ grow without bound as bandwidth $B$ increases \cite{Porrat}.

In the following, we aim to analyze the tradeoff between spectral efficiency and energy efficiency by exploring the minimum bit energy and wideband slope for case I ($B_{\mathrm{c}} \rightarrow \infty$) in the scenario with sparse multipath fading. In particular, for a given QoS exponent $\mu$ and a given average transmit power $P$, we first express the normalized effective capacity (bit/s/Hz) for the considered system in \eqref{normalized_C_E} merely as a function of $B_{\mathrm{c}}$, which is given by
\begin{align}
  \label{normalized_C_E_wideband}
  &C_{\mathrm{E}} \big( B_{\mathrm{c}} \big) \notag\\
  & = - \frac{1}{\mu T B_{\mathrm{c}}} {\mathrm{ln}} \Big( \mathbb{E} \Big\{ {\mathrm{exp}} \Big( {- \mu T B_{\mathrm{c}} {\mathrm{log}}_2 \big( 1 + \frac{P}{N_0 N_{\mathrm{c}} B_{\mathrm{c}}} \xi \big) } \Big) \Big\} \Big).
\end{align}
Then, we can obtain the minimum bit energy and the wideband slope for case I in the scenario with sparse multipath fading by the following theorem.
\begin{Thm}
  \label{wideband}
  In the scenario with sparse multipath fading where the numbers of the independent resolvable subchannels remains constant while the bandwidth of each subchannel increases sublinearly with increasing system bandwidth, i.e., $B_{\mathrm{c}} \rightarrow \infty$, the minimum bit energy and the wideband slope for the considered system in the wideband regime are written as
  \begin{align}
  \label{minimum_bit_energy_wideband}
  \tilde {\frac{E_{\mathrm{b}}}{N_0}}_{\min} &= - \frac{\displaystyle{\frac{\mu T P}{N_0 N_{\mathrm{c}} }}}{{\mathrm{ln}} \Big( \mathbb{E} \Big\{ {\mathrm{exp}} \Big( {- \displaystyle{\frac{\mu T P}{N_0 N_{\mathrm{c}} {\mathrm{ln}}2 } } \xi} \Big) \Big\} \Big)} {\text{ and}} \\
  \label{wideband_slope_wideband}
  \tilde S_0 &= 2 \Big({\displaystyle{\frac{N_0 N_{\mathrm{c}} {\mathrm{ln}}2 }{\mu T P} }}\Big)^2 \mathbb{E} \Big\{ {\mathrm{exp}} \Big( {- \displaystyle{\frac{\mu T P}{N_0 N_{\mathrm{c}} {\mathrm{ln}}2 } } \xi} \Big) \Big\} \notag\\
  &\ \ \ \times \bigg( {{\mathrm{ln}} \Big( \mathbb{E} \Big\{ {\mathrm{exp}} \Big( {- \displaystyle{\frac{\mu T P}{N_0 N_{\mathrm{c}} {\mathrm{ln}}2 } } \xi} \Big) \Big\} \Big)} \bigg)^2 \notag\\
  &\ \ \ \times \bigg( {\mathbb{E} \Big\{ \xi^2 {\mathrm{exp}} \Big( {- \displaystyle{\frac{\mu T P}{N_0 N_{\mathrm{c}} {\mathrm{ln}}2 } } \xi} \Big) \Big\}} \bigg)^{-1},
  \end{align}
  respectively.
\end{Thm}
\begin{proof}
  Please refer to Appendix E for details.
\end{proof}

\begin{Remark}
  \label{remark_bounded_subchannels}
  Note that the results of the minimum bit energy for the considered system acquired in \textbf{Theorem \ref{low_power}} and in \textbf{Theorem \ref{wideband}}, i.e., ${\bar {\frac{E_{\mathrm{b}}}{N_0}}_{\min}}$ in \eqref{minimum_bit_energy_lowpower} and ${\tilde {\frac{E_{\mathrm{b}}}{N_0}}_{\min}}$ in \eqref{minimum_bit_energy_wideband}, respectively, are not identical. Specifically, we have
  \begin{align}
  \tilde {\frac{E_{\mathrm{b}}}{N_0}}_{\min} &= - \frac{\displaystyle{\frac{\mu T P}{N_0 N_{\mathrm{c}} }}}{{\mathrm{ln}} \Big( \mathbb{E} \Big\{ {\mathrm{exp}} \Big( {- \displaystyle{\frac{\mu T P}{N_0 N_{\mathrm{c}} {\mathrm{ln}}2 } } \xi} \Big) \Big\}  \Big)} \notag\\
  & \overset{(g)} \geqslant - \frac{\displaystyle{\frac{\mu T P}{N_0 N_{\mathrm{c}} }}}{\mathbb{E} \Big\{ {\mathrm{ln}} \Big( {\mathrm{exp}} \Big( {- \displaystyle{\frac{\mu T P}{N_0 N_{\mathrm{c}} {\mathrm{ln}}2 } } \xi} \Big) \Big\} \Big)} = \frac{{\mathrm{ln}}2}{\mathbb{E} \big\{ \xi \big\} } = \bar {\frac{E_{\mathrm{b}}}{N_0}}_{\min},
  \end{align}
  where ($g$) is acquired by applying the Jensen's inequality. According to the above inequality, we easily conclude that the considered system for sparse multipath fading may incur more extra energy requirements in the existence of statistical QoS limitations compared with the one for rich multipath fading, owing to multipath sparsity. As such, we conclude that considering the effective capacity potentially incurs various results at the low-power and wideband regimes, especially in the scenario with sparse multipath fading.
\end{Remark}

\begin{Remark}
  \label{wideband_no_QoS}
  By assuming that there are no QoS limitations, i.e., $\mu \rightarrow 0$, the minimum bit energy in \eqref{minimum_bit_energy_wideband} and the wideband slope in \eqref{wideband_slope_wideband} for the considered system in the wideband regime can be further expressed as
  \begin{align}
  \label{tilde_minimum_bit_energy_wideband}
  \tilde {\frac{E_{\mathrm{b}}}{N_0}}_{\min} &= \lim \limits_{\mu \rightarrow 0} \frac{\displaystyle{\frac{\mu T P}{N_0 N_{\mathrm{c}} }}}{- {\mathrm{ln}} \Big( \mathbb{E} \Big\{ {\mathrm{exp}} \Big( {- \displaystyle{\frac{\mu T P}{N_0 N_{\mathrm{c}} {\mathrm{ln}}2 } } \xi} \Big) \Big\} \Big)} \notag\\
  &\overset{(h)} = \lim \limits_{\mu \rightarrow 0} \frac{\displaystyle{\frac{\mu T P}{N_0 N_{\mathrm{c}} }}}{ {\displaystyle{\frac{\mu T P}{N_0 N_{\mathrm{c}} {\mathrm{ln}}2}}} \mathbb{E} \big\{ \xi \big\} + o \Big( \displaystyle{\frac{\mu T P}{N_0 N_{\mathrm{c}} }} \Big) } \notag\\
  & = \lim \limits_{\mu \rightarrow 0} \Bigg( {\frac{ \mathbb{E} \big\{ \xi \big\}}{{\mathrm{ln}} 2} + \frac{o \Big( \displaystyle{\frac{\mu T P}{N_0 N_{\mathrm{c}} }} \Big)}{\displaystyle{\frac{\mu T P}{N_0 N_{\mathrm{c}} }}} } \Bigg)^{-1} \notag\\
  & = \Bigg( {\frac{ \mathbb{E} \big\{ \xi \big\}}{{\mathrm{ln}} 2} + \lim \limits_{\mu \rightarrow 0} \frac{o \Big( \displaystyle{\frac{\mu T P}{N_0 N_{\mathrm{c}} }} \Big)}{\displaystyle{\frac{\mu T P}{N_0 N_{\mathrm{c}} }}} } \Bigg)^{-1} \overset{(i)} = \frac{{\mathrm{ln}}2}{\mathbb{E} \big\{ \xi \big\} }
  \end{align}
  and \eqref{tilde_wideband_slope_wideband} at the top of next page, respectively, where ($h$) and ($j$) can be obtained via exploiting the first-order Taylor series expansion of term ${- {\mathrm{ln}} \big( \mathbb{E} \big\{ {\mathrm{exp}} \big( {- {\frac{\mu T P}{N_0 N_{\mathrm{c}} {\mathrm{ln}}2 } } \xi} \big) \big\} \big)}$ for $\frac{\mu T P}{N_0 N_{\mathrm{c}} } \rightarrow 0$, while ($i$) and ($k$) can be acquired by the fact that $\lim \nolimits_{x \rightarrow 0} \frac{o ( x )}{x} = 0$. Note that we obtain the same results for the considered system as shown in \textbf{Remark \ref{low_power_no_QoS}}. As such, we conclude that the energy efficiency for the considered system can approach that for a system without QoS limitations in the wideband regime.
  \begin{figure*}[ht]
  \begin{align}
  \label{tilde_wideband_slope_wideband}
  \tilde S_0 &= \lim \limits_{\mu \rightarrow 0} 2 \Big({\displaystyle{\frac{N_0 N_{\mathrm{c}} {\mathrm{ln}}2 }{\mu T P} }}\Big)^2 {\mathbb{E} \Big\{ {\mathrm{exp}} \Big( {- \displaystyle{\frac{\mu T P}{N_0 N_{\mathrm{c}} {\mathrm{ln}}2 } } \xi} \Big) \Big\} \bigg( {{\mathrm{ln}} \Big( \mathbb{E} \Big\{ {\mathrm{exp}} \Big( {- \displaystyle{\frac{\mu T P}{N_0 N_{\mathrm{c}} {\mathrm{ln}}2 } } \xi} \Big) \Big\} \Big)} \bigg)^2} \bigg( {\mathbb{E} \Big\{ \xi^2 {\mathrm{exp}} \Big( {- \displaystyle{\frac{\mu T P}{N_0 N_{\mathrm{c}} {\mathrm{ln}}2 } } \xi} \Big) \Big\}} \bigg)^{-1} \notag\\
  &\overset{(j)} = \lim \limits_{\mu \rightarrow 0} 2 \Big({\displaystyle{\frac{N_0 N_{\mathrm{c}} {\mathrm{ln}}2 }{\mu T P} }}\Big)^2 \frac{ \Big( {\displaystyle{\frac{\mu T P}{N_0 N_{\mathrm{c}} {\mathrm{ln}}2}}} \mathbb{E} \big\{ \xi \big\} + o \Big( \displaystyle{\frac{\mu T P}{N_0 N_{\mathrm{c}} }} \Big) \Big)^2}{\mathbb{E} \big\{ \xi^2 \big\}} = \frac{2 \big( {\mathrm{ln}}2 \big)^2}{\mathbb{E} \big\{ \xi^2 \big\}} \Bigg( {\frac{ \mathbb{E} \big\{ \xi \big\}}{{\mathrm{ln}} 2} + \lim \limits_{\mu \rightarrow 0} \frac{o \Big( \displaystyle{\frac{\mu T P}{N_0 N_{\mathrm{c}} }} \Big)}{\displaystyle{\frac{\mu T P}{N_0 N_{\mathrm{c}} }}} } \Bigg)^2 \overset{(k)} = \frac{2 \big( \mathbb{E} \big\{ \xi \big\} \big)^2 }{\mathbb{E} \big\{ \xi^2 \big\}}.
  \end{align}\hrulefill\vspace*{-4mm}
  \end{figure*}

\end{Remark}

Next, we consider a special scenario for the considered system with the sufficiently large number of the IRS reflecting elements in the wideband regime. Specifically, we obtain the minimum bit energy and the wideband slope in this scenario by the following corollary.
\begin{Cor}
  \label{wideband_large_number_elements}
  By assuming that the number of the IRS reflecting elements is sufficiently large, the minimum bit energy in \eqref{minimum_bit_energy_wideband} and the wideband slope in \eqref{wideband_slope_wideband} for the considered system in the wideband regime are further approximated by
  \begin{align}
  \label{minimum_bit_energy_large_number_elements_wideband}
  \acute {\frac{E_{\mathrm{b}}}{N_0}}_{\min} &= \frac{\displaystyle{\frac{\mu T P}{N_0 N_{\mathrm{c}} }}}{\alpha_{\xi }{\mathrm{ln}} \Big( {\displaystyle{\frac{\mu T P \beta_{\xi}}{N_0 N_{\mathrm{c}} {\mathrm{ln}}2}} + 1} \Big)} {\text{ and}} \\
  \label{wideband_slope_large_number_elements_wideband}
  \acute S_0 &= 2 \frac{\alpha_{\xi} \Big( {\beta_{\xi} + {\displaystyle{\frac{N_0 N_{\mathrm{c}} {\mathrm{ln}}2 }{\mu T P} }}} \Big)^2}{\big( \alpha_{\xi} + 1 \big) \beta_{\xi}^2} \Big( {\mathrm{ln}} \Big( {\displaystyle{\frac{\mu T P \beta_{\xi}}{N_0 N_{\mathrm{c}} {\mathrm{ln}}2}} + 1} \Big) \Big)^2,
  \end{align}
  respectively, where $\alpha_{\xi}$ and $\beta_{\xi}$ are the shape and scaling parameters of a Gamma distribution for $\xi$, respectively.
\end{Cor}
\begin{proof}
  Please refer to Appendix F for details.
\end{proof}

\begin{Remark}
  \label{remark_infiniteN_wideband}
  By considering a special case with the infinite large number of the IRS reflecting elements, i.e., $N \rightarrow \infty$, we can further approximate the minimum bit energy in \eqref{minimum_bit_energy_large_number_elements_wideband} and the wideband slope in \eqref{wideband_slope_large_number_elements_wideband} for the considered system in the wideband regime as
  \begin{align}
  \acute {\frac{E_{\mathrm{b}}}{N_0}}_{\min} = 0 {\text{ and }} \acute S_0 = \infty,
  \end{align}
  respectively. Different from the results for $N \rightarrow \infty$ in the low-power regime obtained in \textbf{Remark \ref{remark_infiniteN_lower_power}}, infinite spectral efficiency performance can be achieved at infinitesimal bit energy levels in the wideband regime where the number of the independent resolvable subchannels is bounded.
\end{Remark}

Based on the above special scenario for the considered system, we further consider a practical implementation of the discrete phase shift with $b$ bit resolution. Similarly, the channel coefficient gain can be expressed as $\bar \xi = \big( \sqrt{\ell_{\mathrm{h}}} \xi_{\mathrm{d}} + \sqrt{\ell_{\mathrm{f}} \ell_{\mathrm{g}}} \bar \xi_{\mathrm{r}} \big)^2$, where $\bar \xi_{\mathrm{r}} = \sum\nolimits_{n\in {\mathcal{N}}} \big| f_n \big| \big| g_n \big| {\mathrm{exp}} \big({\jmath{\hat \theta_n}} \big)$, where $\hat \theta_n$ is the quantization error of the discrete phase shift of the $n$-th IRS reflecting element \cite{Qingqing2}. As such, the minimum bit energy and the wideband slope under the assumption of the discrete phase shifts can be obtained by the following corollary.
\begin{Cor}
  \label{discrete_phase_shifts_wideband}
  By assuming that the number of the IRS reflecting elements, i.e., $N$, is sufficiently large while the quantization error of the discrete phase shift of the $n$-th IRS reflecting element, i.e., $\hat \theta_n$, is independently and uniformly distributed in $\big[ {- \pi/2^b,\pi/2^b} \big)$, the minimum bit energy in \eqref{minimum_bit_energy_large_number_elements_wideband} and the wideband slope in \eqref{wideband_slope_large_number_elements_wideband} for the considered system in the wideband regime are further approximated by $\grave {\frac{E_{\mathrm{b}}}{N_0}}_{\min}$ and $\grave S_0$, respectively, which share the similar structures with $\acute {\frac{E_{\mathrm{b}}}{N_0}}_{\min}$ in \eqref{minimum_bit_energy_large_number_elements_wideband} and $\acute S_0$ in \eqref{wideband_slope_large_number_elements_wideband}.
\end{Cor}
\begin{proof}
  The proof of the corollary similarly follows Appendix D. Here, we omit the details of the proof due to the space limitation of this paper.
\end{proof}

\begin{Remark}
  Note that the above results of the minimum bit energy and the wideband slope for the considered system with the discrete nature of IRS reflecting elements in the wideband regime can gradually approach the results in \eqref{minimum_bit_energy_large_number_elements_wideband} and \eqref{wideband_slope_large_number_elements_wideband} obtained in \textbf{Corollary \ref{wideband_large_number_elements}}, respectively, as the bit resolution $b \rightarrow \infty$ that represents the continuous phase shifts without quantization errors, which can be explained by the fact that $\lim \nolimits_{b \rightarrow \infty} \frac{2^b}{\pi} \sin \big( \frac{\pi}{2^b} \big) = 1$.
\end{Remark}

The above analysis holds for case I where the number of the independent resolvable subchannels is bounded, $B_{\mathrm{c}}$ grows linearly while $N_{\mathrm{c}}$ remains constant with increasing bandwidth $B$. Now, we analyze energy efficiency for case II where the number of the independent resolvable subchannels grows but only sublinearly, both $B_{\mathrm{c}}$ and $N_{\mathrm{c}}$ increase without bound as system bandwidth $B$ increases. In particular, we can obtain the minimum bit energy and the wideband slope for case II in the scenario with sparse multipath fading by the following theorem.
\begin{Thm}
  \label{N_c_infty_wideband}
  In the scenario with sparse multipath fading where the number of the independent resolvable subchannels and the bandwidth of each subchannel increase sublinearly with increasing system bandwidth, i.e., $B_{\mathrm{c}} \rightarrow \infty$ and $N_{\mathrm{c}} \rightarrow \infty$, respectively, the minimum bit energy and the wideband slope for the considered system in the wideband regime, respectively, are expressed as
  \begin{align}
  \label{B_c_N_c_infty}
  \tilde {\frac{E_{\mathrm{b}}}{N_0}}_{\min} = \frac{{\mathrm{ln}}2}{\mathbb{E} \big\{ \xi \big\} } {\text{ and }} \tilde S_0 &= \frac{2 \big( \mathbb{E} \big\{ \xi \big\} \big)^2 }{\mathbb{E} \big\{ \xi^2 \big\}}.
  \end{align}
\end{Thm}
\begin{proof}
  The minimum bit energy and the wideband slope can be rewritten by employing $N_{\mathrm{c}} \rightarrow \infty$ in the results of \textbf{Theorem \ref{wideband}} and further simplified via adopting the similar method as shown in \textbf{Remark \ref{wideband_no_QoS}}. Here, we omit the details of the proof due to the space limitation of this paper.
\end{proof}

\begin{Remark}
  In the wideband regime, the expressions of the minimum bit energy for rich multipath fading where the number of subchannels increases linearly and sparse multipath fading where the number of subchannels and the bandwidth of each subchannel increase sublinearly with increasing bandwidth (i.e., case II) are identical. Note that the number of the independent resolvable subchannels is regarded as a measure of the degrees-of-freedom (DoFs) for the considered system. On the one hand, the value of the minimum bit energy for sparse multipath fading where the number of subchannels remains constant as the system bandwidth increases (i.e., case I) is bounded when the DoFs is bounded with increasing bandwidth as shown in \textbf{Remark \ref{remark_bounded_subchannels}}. On the other hand, when the DoFs grow with the system bandwidth, the value of the minimum bit energy for sparse multipath fading (i.e., case II) is same as that for rich multipath fading, despite the fact that the number of subchannels grows sublinearly.
\end{Remark}

\begin{Remark}
  Note that the minimum bit energy and the wideband slope acquired in \textbf{Theorem \ref{N_c_infty_wideband}} for $N_{\mathrm{c}} \rightarrow \infty$ and fixed $\mu$ is equivalent to that analyzed in \textbf{Remark \ref{wideband_no_QoS}} for $\mu \rightarrow 0$ and fixed $N_{\mathrm{c}}$. In other words, the energy efficiency gains can be achieved by exploiting the results in \eqref{B_c_N_c_infty} when $\mu = 0$, even though the DoFs provided by the independent resolvable subchannels is bounded with increasing bandwidth. As such, we conclude that the performance for the considered system unaffected by the multipath sparsity can approach that for a system in sparse multipath fading without QoS limitations.
\end{Remark}

\section{Simulation Results}

\subsection{Simulation Setup}

\begin{figure}[t]
\centering
\includegraphics[width=3in]{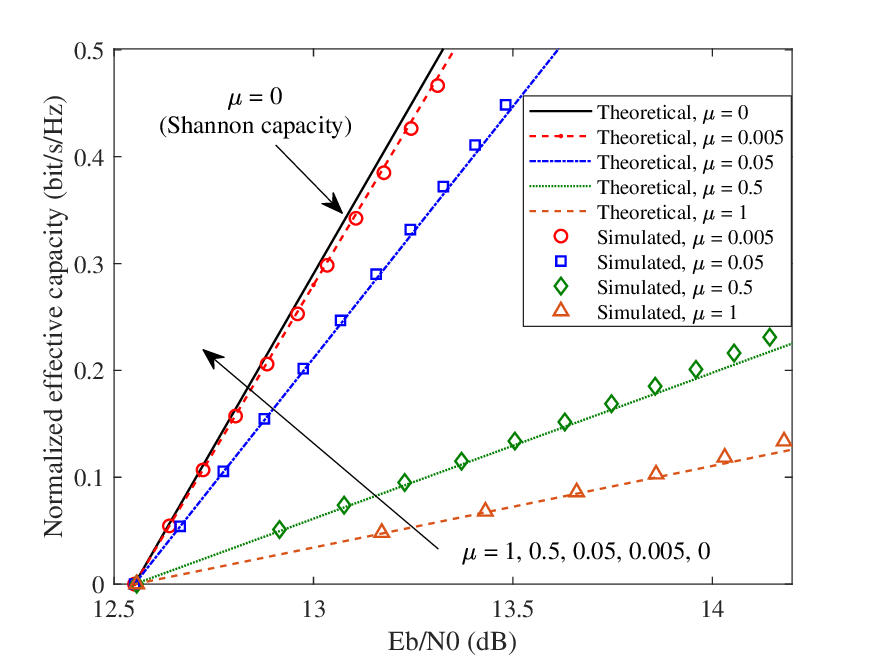}
\caption{The normalized effective capacity (bit/s/Hz) versus the bit energy (dB) for the considered system in the low-power regime with $N=100$ and $B = 10^5$ Hz.}\vspace{-1em}
\label{fig:lowpower_IRS}
\end{figure}

In this section, we aim to numerically demonstrate the above analytical results. For simulation, we consider an IRS-aided wireless communication system, where a transmitter, a receiver, and an IRS are located at $\left( 0,0 \right)$, $\left( 10,0 \right)$ m, and  $\left( 5,10 \right)$ m, respectively. For the large-scale fading, the distance-dependent path loss model is expressed as $\ell = {L_0}d^{ - \vartheta}$, where ${L_0}$, $d$, and $\vartheta$ represent the path loss at the reference distance of $1$ m, the individual link distance, and the corresponding path loss exponent, respectively. Specifically, the path loss exponent for the transmitter-receiver link is set to $3.6$, while the counterparts for the transmitter-IRS and IRS-receiver links are set to $2.2$. Also, ${L_0}$ is set to $-30$ dB. For the small-scale fading, the fading coefficients of ${\mathbf{g}}$, ${\mathbf{f}}$, and $h$ follow Nakagami-$m$ distribution with the corresponding shape parameters, i.e., $m_{\mathrm{g}} = 1$, $m_{\mathrm{f}} = 1$, and $m_{\mathrm{h}} = 2$, respectively. The number of the channel realizations is set to $10^3$. The block duration is set to $T = 2$ ms.

\subsection{Spectral Efficiency versus Energy Efficiency in Low-Power Regime}

Fig. \ref{fig:lowpower_IRS} demonstrates the normalized effective capacity versus the bit energy for the considered system in the low-power regime. As expected, all the curves have the same value of the minimum bit energy that is not effected by the statistical QoS limitations, which is in alignment with the analysis result in \textbf{Theorem \ref{low_power}}. Note that the minimum bit energy is large due to the existence of the large-scale fading. Furthermore, we observe that the effective capacity is identical to the traditional Shannon capacity when there are no QoS limitations, i.e., $\mu=0$. As can be observed, for the given energy requirement, the normalized effective capacity gradually decreases as the queueing constraints become more strict. This can be intuitively explained by the fact that the arrival rates supported by the considered system decline with the increasingly more stringent buffer constraints, which leads to lower departure rates. Besides, it is noted that the wideband slope for the considered system in the low-power regime diminishes as $\mu$ grows. This is because the second derivative of the normalized effective capacity at ${\mathrm{SNR}} = 0$ in \eqref{ddot_C_E} is a monotonically decreasing function with increasing QoS exponent $\mu$. As such, we conclude that the existence of QoS limitations can incur spectral efficiency degradation or more additional energy cost at low but nonzero SNR levels.

\begin{figure}[t]
\centering
\includegraphics[width=3in]{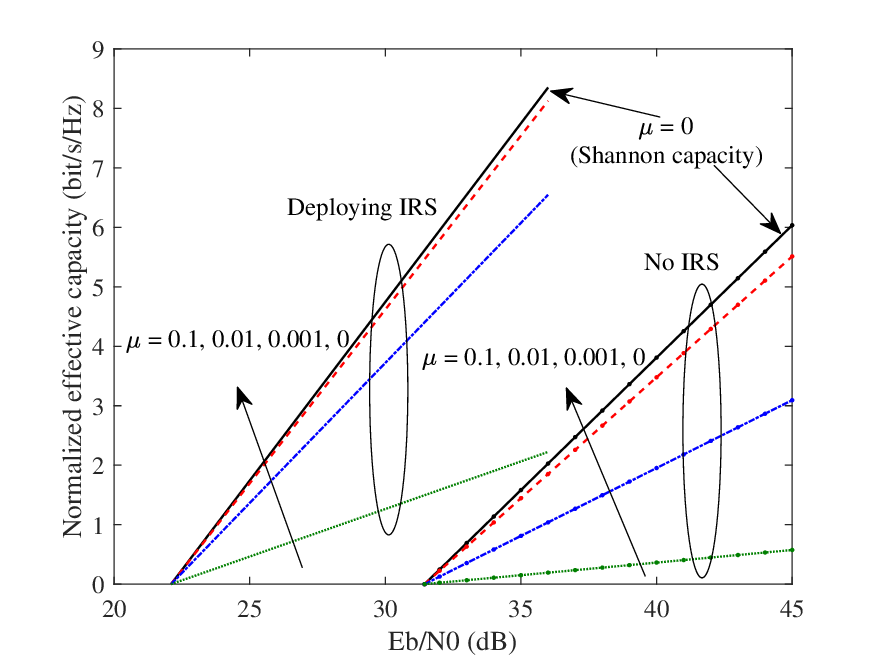}
\caption{The normalized effective capacity (bit/s/Hz) versus the bit energy (dB) for comparison between the considered system and the non-IRS system in the low-power regime with $N=60$ and $B = 10^5$ Hz.}\vspace{-1em}
\label{fig:lowpower_no_IRS}
\end{figure}

Fig. \ref{fig:lowpower_no_IRS} studies the normalized effective capacity versus the bit energy for comparison between the considered system and the non-IRS system in the low-power regime. We observe that the minimum bit energy for the considered system is smaller than that for the non-IRS system, no matter whether there is the QoS limitations or not. This indicates the superiority of improving the system energy efficiency by introducing IRS. Furthermore, we note that the wideband slope for the considered system is larger than that for the non-IRS system under the same QoS constraints. This can achieve significant system performance gaps between these two systems in terms of energy saving and spectral efficiency improvement, especially under the strict QoS limitations. For instance, for achieving $0.5$ bit/s/Hz/ of spectral efficiency under $\mu = 0.1$, the considered system only requires $25.5$ dB of bit energy, while the non-IRS system requires roughly $43.5$ dB. The corresponding performance gap is tremendous thanks to the higher power gains provided by IRS. On the other hand, the IRS-aided system can achieve higher performance gain by $93 \%$ compared with the non-IRS system when the bit energy requirement is $35$ dB under $\mu = 0.1$, due to exploiting the extra DoFs introducing by IRS.  As such, we conclude that deploying IRS is able to compensate the spectral efficiency degradation or fulfil stringent energy requirements caused by the strict QoS constraints at low but nonzero SNR levels.

\begin{figure}[t]
\centering
\includegraphics[width=3in]{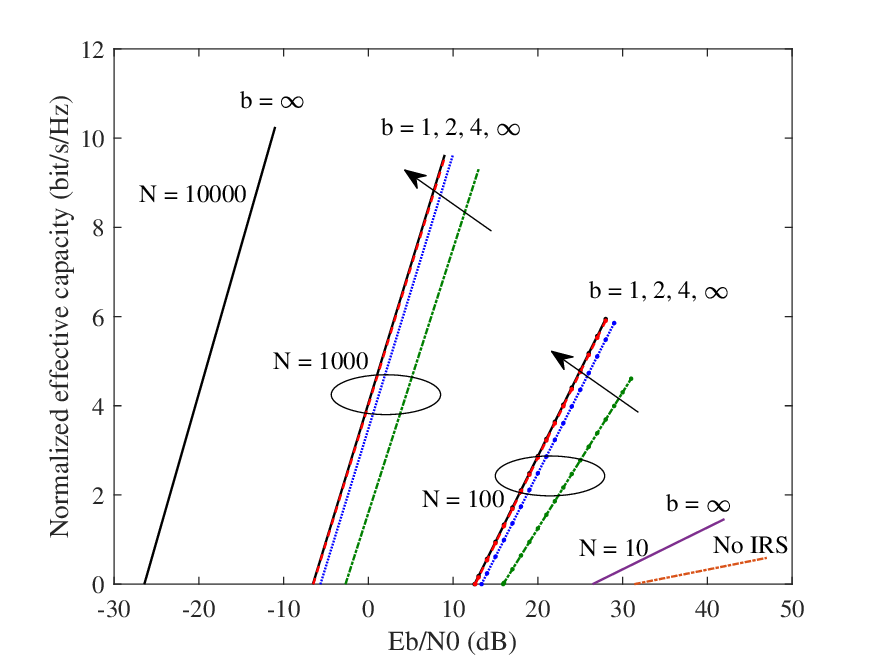}
\caption{The normalized effective capacity (bit/s/Hz) versus the bit energy (dB) for the considered system by deploying the various number of the IRS reflecting elements in the low-power regime with $\mu=0.1$ and $B = 10^5$ Hz.}\vspace{-1em}
\label{fig:lowpower_IRS_element}
\end{figure}

Fig. \ref{fig:lowpower_IRS_element} reveals the normalized effective capacity versus the bit energy for the considered system by deploying the various number of the IRS reflecting elements in the low-power regime. Note that the value of the minimum bit energy for the considered system decreases gradually with the increasing number of the IRS reflecting elements, as the IRS can reflect additional signal power received from the transmitter facilitating the acquisition of  a higher power gain. Furthermore, it is observed that the wideband slope for the considered system increase monotonically in the growth of the number of IRS elements, since the system can exploit the additional DoFs offered by the more IRS reflecting elements to compensate the spectral efficiency degradation caused by the strict QoS constraints. More importantly, the wideband slope becomes saturated to $2$ bit/s/Hz/($3$ dB) when the number of the IRS reflecting elements is sufficiently large, which is in alignment with the analysis result in \textbf{Remark \ref{remark_infiniteN_lower_power}}. Moreover, we note that the system performance of the IRS with a limited bit resolution is able to approach that of the upper bound achieved by adopting an infinite bit resolution. Besides, as can be observed, the effect of the quantization error of the phase shifts on the wideband slope for the considered system alleviates in the case of the excessive number of IRS reflecting elements. As such, we conclude that deploying the sufficiently large number of the practical IRS reflecting elements can significantly decline the energy consumption for the required spectral efficiency performance in the low-SNR regime, even with limited bit-resolution phase shifters.

\subsection{Spectral Efficiency versus Energy Efficiency in Wideband Regime}

\begin{figure}[t]
\centering
\includegraphics[width=3in]{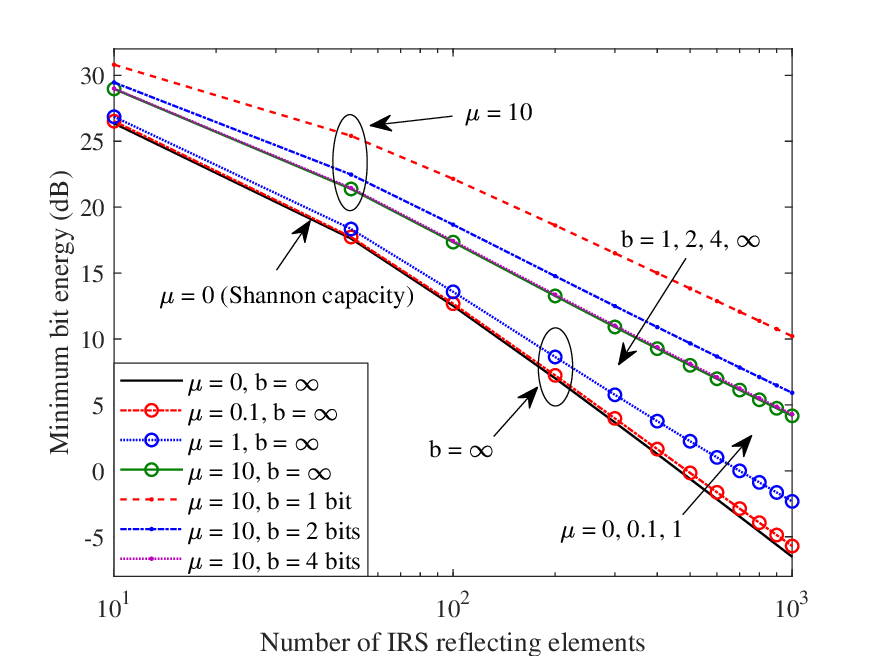}
\caption{The minimum bit energy (dB) versus the number of IRS reflecting elements for case I in the scenario with sparse multipath fading in the wideband regime with $N_{\mathrm{c}}=5$ and $P/N_0 = 10^6$.}\vspace{-1em}
\label{fig:wideband_IRS_element}
\end{figure}

\begin{figure*}[t]
\centering
\subfigure[case I with $N_{\mathrm{c}}=5$ and increasing $B_{\mathrm{c}}$.]{
\label{fig:wideband_case_I}
\includegraphics[width=3in]{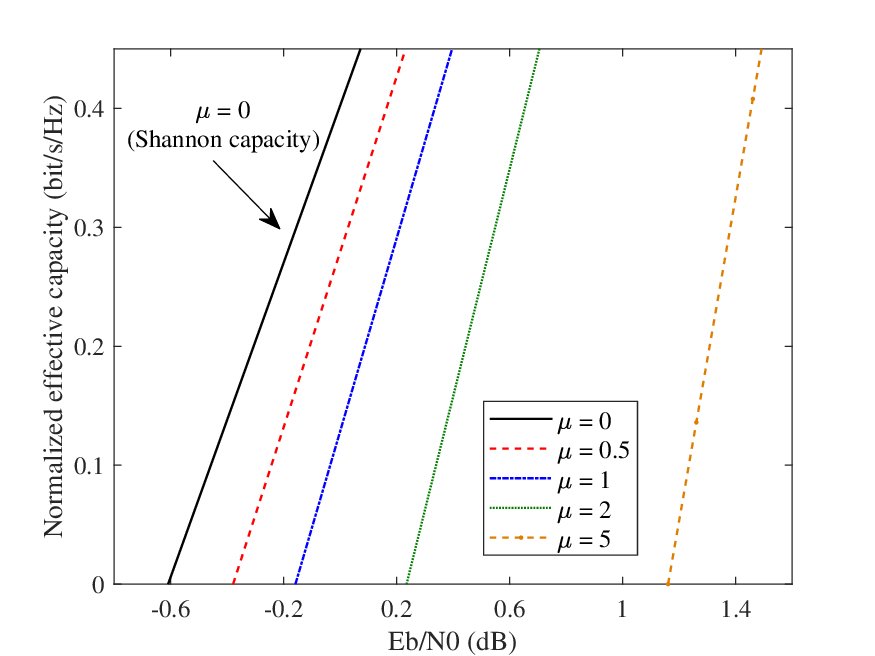}}
\hspace{0.5in}
\subfigure[case II with increasing $N_{\mathrm{c}}$ and $B_{\mathrm{c}}$.]{
\label{fig:wideband_case_II}
\includegraphics[width=3in]{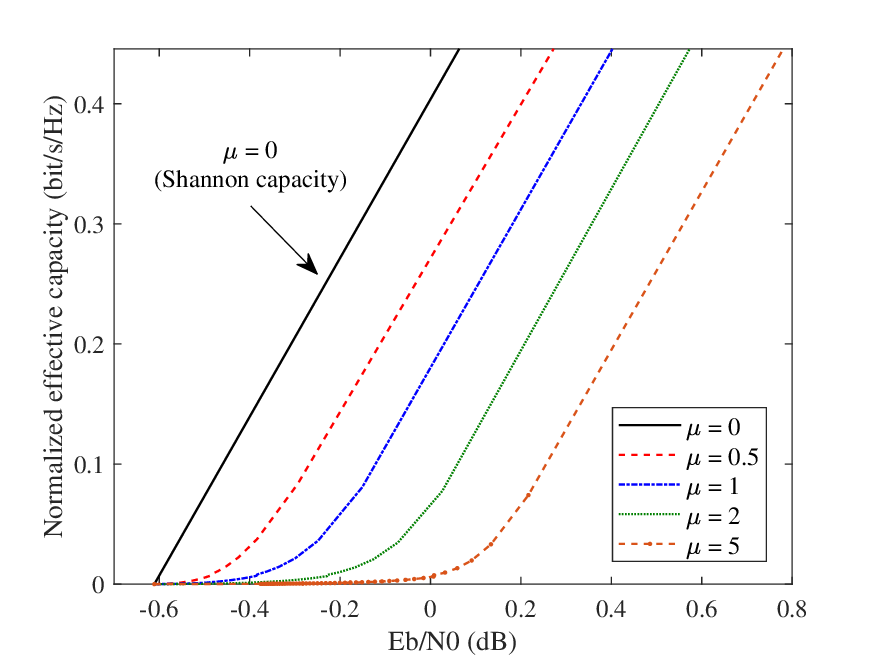}}
\caption{\footnotesize \setstretch{1.5}The normalized effective capacity (bit/s/Hz) versus the bit energy (dB) for a) case I and b) case II in the scenario with sparse multipath fading in the wideband regime with $N=500$ and $P/N_0 = 10^5$.}\vspace{-0.5em}
\end{figure*}

Fig. \ref{fig:wideband_IRS_element} studies the minimum bit energy versus the number of IRS reflecting elements for case I in the scenario with sparse multipath fading in the wideband regime, where the numbers of the independent resolvable subchannels remains constant, i.e., $N_{\mathrm{c}}=5$, while the bandwidth of each subchannel, i.e., $B_{\mathrm{c}}$, increases sublinearly with increasing system bandwidth. We observe that the minimum bit energy requirements for the considered system become relaxed gradually with the increasing number of IRS reflecting elements, since introducing more extra reflecting elements of the IRS can achieve the higher power gains at the user. Additionally, we observe that the minimum bit energy for the considered system grows gradually with increasing QoS exponent $\mu$, due to the existence of more additional energy consumption incurred by more strict QoS constraints. Besides, for given QoS constraints, the minimum bit energy requirement for the considered system with the deployment of the limited bit-resolution phase shifters can be more stringent in the wideband regime compared to that with implementing the continuous phase shifts. This is because the discrete nature of the phase shifts may lead to substantial energy penalty for the required effective capacity in the scenario with sparse multipath fading.

Fig. \ref{fig:wideband_case_I} demonstrates the normalized effective capacity versus the bit energy for case I in the scenario with sparse multipath fading in the wideband regime, where the numbers of the independent resolvable subchannels remains constant, i.e., $N_{\mathrm{c}}=5$, while the bandwidth of each subchannel, i.e., $B_{\mathrm{c}}$, increases sublinearly with increasing system bandwidth. As can be observed, the minimum bit energy for the considered system grows gradually with increasing QoS exponent $\mu$, which is different from the simulation results as shown in Fig. \ref{fig:lowpower_IRS}. This indicates that the considered system for sparse multipath fading may incur more extra energy requirements in the existence of statistical QoS constraints compared with the one for rich multipath fading, owing to multipath sparsity. Moreover, we note that the wideband slope for the considered system increases monotonically as $\mu$ grows, which means that the increment of the bit energy required for improving the given system performance is more slightly when $\mu$ is larger. Yet, the absolute bit energy requirements for a given spectral efficiency are still strict, especially under the stringent QoS limitations. Besides, it is noted that when the bit-resolution of the phase shifters is limited, the minimum bit energy for the considered system becomes large in the wideband regime. As such, we conclude that the existence of the quantization error of the phase shifts may incur significant energy penalty for the required spectral efficiency performance at low SNR levels, especially in the existence of the strict QoS limitations.

Fig. \ref{fig:wideband_case_II} unveils the normalized effective capacity versus the bit energy for case II in the scenario with sparse multipath fading in the wideband regime, where the number of the independent resolvable subchannels and the bandwidth of each subchannel increase sublinearly with increasing system bandwidth, i.e., $N_{\mathrm{c}}$ increases from $5$ to $50$ as $B_{\mathrm{c}}$  grows from $10$ kHz to $10$ MHz, respectively. As expected, all the curves have the same value of the minimum bit energy for sparse multipath fading that is unaffected by the statistical QoS limitations, which is same as that for rich multipath fading as shown in \textbf{Theorem \ref{N_c_infty_wideband}}, thanks to the DoFs provided by the substantial independent resolvable subchannels with the increasing system bandwidth. Nonetheless, it is observed that approaching the minimum bit energy in the curves with $\mu > 0$ is slow due to the zero wideband slopes, which unveils that it is highly challenging to attain the value of the minimum bit energy by increasing system bandwidth as long as there are QoS limitations in IRS-aided wireless systems.

\section{Conclusion}

In this paper, we analyzed the energy efficiency of the IRS-aided wireless communication systems under certain statistical QoS constraints. Specifically, we employed the effective capacity formulation to measure the throughput performance for the considered system at the bit energy levels. We determined the expressions of the minimum bit energy and investigated the tradeoff between spectral efficiency and energy efficiency in the low-power and wideband regimes under such QoS constraints. In the low-power regime, we demonstrated that the minimum bit energy requirement for the considered system under statistical QoS limitations is identical to that required in the absence of such limitations. Also, we showed that deploying the sufficiently large number of the practical IRS reflecting elements can significantly diminish the energy consumption for the required spectral efficiency performance at low but nonzero SNR levels, even with limited bit-resolution phase shifters. In the wideband regime, we revealed that the results acquired in the low-power regime is also applicable to the scenario with rich multipath fading. Moreover, we proved that the expression of the minimum bit energy for rich multipath fading is identical to that for sparse multipath fading where the number of subchannels grows sublinearly with the increasing bandwidth. Besides, we unveiled that compared with the results applied to the low-power regime, higher effective capacity performance can be achieved in the scenario with sparse multipath fading where the number of the independent resolvable subchannels is bounded while satisfying the same minimum bit energy requirement in the wideband regime.

\appendix
\begin{appendices}

\subsection{Proof of Lemma \ref{optimal_theta}}

By fixing the transmit SNR, the optimization problem in \eqref{C_E} can be simplified to
\begin{align}
  \label{p1}
  &\mathop {{\mathrm{maximize}}}\limits_{\mathbf{\Theta}} {\text{  }} \big| \sqrt{\ell_{\mathrm{f}} \ell_{\mathrm{g}}} {\mathbf{f}}^T {\mathbf{\Theta }} {\mathbf{g}} + \sqrt{\ell_{\mathrm{h}}} h \big|^2 \notag\\
  &{\mathrm{s.t.}}{\text{  }} \big| \theta_m \big|=1, \ \forall m \in {\mathcal{M}}.
\end{align}
Note that the problem in \eqref{p1} admits a closed-form solution thanks to the special structure of the objective function that satisfies the following inequality, which is given by
\begin{align}
  \label{triangle_inequality}
  \big| \sqrt{\ell_{\mathrm{f}} \ell_{\mathrm{g}}} {\mathbf{f}}^T {\mathbf{\Theta }} {\mathbf{g}} + \sqrt{\ell_{\mathrm{h}}} h \big| \leqslant \sqrt{\ell_{\mathrm{f}} \ell_{\mathrm{g}}} \big| {\mathbf{f}}^T {\mathbf{\Theta }} {\mathbf{g}} \big| + \sqrt{\ell_{\mathrm{h}}} \big| h \big|,
\end{align}
where the corresponding equality holds if and only if $\arg\big( {\mathbf{f}}^T {\mathbf{\Theta }} {\mathbf{g}} \big) = \arg\big( h \big)$ based on the triangle inequality. In the following, we aim to demonstrate that there is always a solution ${\mathbf{\Theta }}$ for problem in \eqref{p1}, which satisfies the corresponding phase shift constraints and the equality in \eqref{triangle_inequality}. In particular, we first rewrite term ${\mathbf{f}}^T {\mathbf{\Theta }} {\mathbf{g}}$ in \eqref{p1} as ${\mathbf{v}}^H {\mathrm{diag}} \big( {\mathbf{f}}^T \big) {\mathbf{g}}$, where ${\mathbf{v}} = \big[ {e^{\jmath{\theta _1}}, \ldots, e^{\jmath{\theta _M}}} \big]^H \in {\mathbb{C}^{M \times 1}}$. Then, the problem in \eqref{p1} can be equivalently reformulated as
\begin{align}
  \label{p2}
  &\mathop {{\mathrm{maximize}}}\limits_{\mathbf{v}} {\text{  }} \big| {\mathbf{v}}^H {\mathrm{diag}} \big( {\mathbf{f}}^T \big) {\mathbf{g}} \big|^2 \notag\\
  &{\mathrm{s.t.}}{\text{  }} \big| v_m \big|=1, \ \forall m \in {\mathcal{M}}, \notag\\
  &\ \ \ \ \ \arg\big( {\mathbf{v}}^H {\mathrm{diag}} \big( {\mathbf{f}}^T \big) {\mathbf{g}} \big) = \arg\big( h \big),
\end{align}
where $v_m$ is the $m$-th element of optimization vector ${\mathbf{v}}$. Obviously, we can easily acquire the optimal solution for the problem in \eqref{p2}, i.e., ${\mathbf{v}}^{\star} = {\mathrm{exp}} \big( \jmath \arg( h ) - \jmath \arg( {\mathrm{diag}} ( {\mathbf{f}}^T ) {\mathbf{g}} ) \big)$. This completes the proof of the proposition.

\subsection{Proof of Theorem \ref{low_power}}

The first and second derivatives of the normalized effective capacity with respect to SNR in \eqref{normalized_C_E} are given by \eqref{dot_C_E_SNR} and \eqref{ddot_C_E_SNR} at the top of next page,
\begin{figure*}[ht]
\begin{align}
  \label{dot_C_E_SNR}
  \dot{C}_{\mathrm{E}} \big( {\mathrm{SNR}} \big) &= \frac{1}{{\mathrm{ln}}2} \bigg( {\mathbb{E} \Big\{ {\mathrm{exp}} \Big( {- \mu T B {\mathrm{log}}_2 \big( 1 + {\mathrm{SNR}} \xi \big) } \Big) \Big\} } \bigg)^{-1} \mathbb{E} \bigg\{ \xi \Big( { 1 + {\mathrm{SNR}} \xi } \Big)^{-1} {\mathrm{exp}} \Big( {- \mu T B {\mathrm{log}}_2 \big( 1 + {\mathrm{SNR}} \xi \big) } \Big) \bigg\}. \\
  \label{ddot_C_E_SNR}
  \ddot{C}_{\mathrm{E}} \big( {\mathrm{SNR}} \big) &= \frac{\mu T B}{( {\mathrm{ln}}2 )^2} \bigg( \mathbb{E} \bigg\{ \xi \Big( { 1 + {\mathrm{SNR}} \xi } \Big)^{-1} {\mathrm{exp}} \Big( {- \mu T B {\mathrm{log}}_2 \big( 1 + {\mathrm{SNR}} \xi \big) } \Big) \bigg\} \bigg)^2 - \bigg( \frac{\mu T B}{( {\mathrm{ln}}2 )^2} + \frac{1}{{\mathrm{ln}}2} \bigg) \notag\\
  & \ \ \ \ \times \bigg( {\mathbb{E} \Big\{ {\mathrm{exp}} \Big( {- \mu T B {\mathrm{log}}_2 \big( 1 + {\mathrm{SNR}} \xi \big) } \Big) \Big\} } \bigg)^{-1} \mathbb{E} \bigg\{ \xi^2 \Big( { 1 + {\mathrm{SNR}} \xi } \Big)^{-2} {\mathrm{exp}} \Big( {- \mu T B {\mathrm{log}}_2 \big( 1 + {\mathrm{SNR}} \xi \big) } \Big) \bigg\}.
\end{align}\hrulefill\vspace*{-4mm}
\end{figure*}
respectively. As such, the first and second derivatives of the normalized effective capacity at ${\mathrm{SNR}} = 0$ are expressed as
\begin{align}
  \label{dot_C_E}
  \dot{C}_{\mathrm{E}} \big( 0 \big) &= \frac{\mathbb{E} \big\{ \xi \big\} }{{\mathrm{ln}}2} {\text{ and}} \\
  \label{ddot_C_E}
  \ddot{C}_{\mathrm{E}} \big( 0 \big) &= \frac{\mu T B}{( {\mathrm{ln}}2 )^2} \Big( \mathbb{E} \big\{ \xi \big\} \Big)^2 - \bigg( \frac{\mu T B}{( {\mathrm{ln}}2 )^2} + \frac{1}{{\mathrm{ln}}2} \bigg) \mathbb{E} \big\{ \xi^2 \big\},
\end{align}
respectively. By substituting \eqref{dot_C_E} into \eqref{min_bit_energy} and substituting \eqref{dot_C_E}--\eqref{ddot_C_E} into \eqref{wideband_slope}, respectively, we acquire the desired results. This completes the proof of the theorem.

\subsection{Proof of Corollary \ref{large_number_elements}}

To further describe the minimum bit energy in \eqref{minimum_bit_energy_lowpower} and the wideband slope in \eqref{wideband_slope_lowpower} for the considered system in the low-power regime, we need to derive the distribution of $\xi = \big( \sqrt{\ell_{\mathrm{h}}} \xi_{\mathrm{d}} + \sqrt{\ell_{\mathrm{f}} \ell_{\mathrm{g}}} \xi_{\mathrm{r}} \big)^2$, where $\xi_{\mathrm{d}} = |h|$ and $\xi_{\mathrm{r}} = \sum\nolimits_{n\in {\mathcal{N}}} |f_n| |g_n|$. In fact, we note that the exact distribution of $\xi$ is hardly acquired in the tractable closed form, since it contains the sum of the products of Nakagami-$m$ random variables, i.e., $\xi_{\mathrm{r}}$. To handle this issue, we aim to approximate the distribution of $\xi_{\mathrm{r}}$ as a Gamma distribution by applying the following lemma.
\begin{Lem}
  \label{moment_matching}
  (Moment Matching \cite{Tahir}) For a non-negative random variable $Y$ with the first moment $\mathbb{E} \big\{ Y \big\}$ and the second moment $\mathbb{E} \big\{ Y^2 \big\}$, the random variable $Y$ can be approximated as a Gamma distribution, i.e., $Y \sim \mathrm{Gamma} \big( \alpha_Y,\beta_Y \big)$, where $\alpha_Y$ and $\beta_Y$ are shape and scaling parameters, respectively, which are given by
  \begin{align}
  \label{shape_scaling}
  \alpha_Y = \frac{\big( \mathbb{E} \big\{ Y \big\} \big)^2}{\mathbb{E} \big\{ Y^2 \big\} - \big( \mathbb{E} \big\{ Y \big\} \big)^2} {\text{ and }} \beta_Y = \frac{\mathbb{E} \big\{ Y^2 \big\} - \big( \mathbb{E} \big\{ Y \big\} \big)^2}{\mathbb{E} \big\{ Y \big\}}.
  \end{align}
\end{Lem}
\begin{proof}
  Please refer to \cite{Tahir} for details.
\end{proof}

For a given sufficiently large $N$, term $\xi_{\mathrm{r}}$ can be regarded as a sum of $N$ independently and identically distributed (i.i.d.) double-Nakagami random variables, i.e., $|f_n| |g_n|$, $\forall n \in {\mathcal{N}}$. As such, we can first express the first and second moments of term $\xi_{\mathrm{r}}$ as \cite{Tahir}
\begin{align}
  \label{xi_r_1st}
  \mathbb{E} \Big\{ \xi_{\mathrm{r}} \Big\} &= \frac{N}{\sqrt{m_{\mathrm{g}} m_{\mathrm{f}}}} \frac{\Gamma \Big( m_{\mathrm{g}} + \displaystyle{\frac{1}{2}} \Big)}{\Gamma \Big( m_{\mathrm{g}} \Big)} \frac{\Gamma \Big( m_{\mathrm{f}} + \displaystyle{\frac{1}{2}} \Big)}{\Gamma \Big( m_{\mathrm{f}} \Big)} {\text{ and}} \\
  \label{xi_r_2rd}
  \mathbb{E} \Big\{ \xi_{\mathrm{r}}^2 \Big\} &= N + \frac{N (N - 1)}{m_{\mathrm{g}} m_{\mathrm{f}}} \frac{\Gamma^2 \Big( m_{\mathrm{g}} + \displaystyle{\frac{1}{2}} \Big)}{\Gamma^2 \Big( m_{\mathrm{g}} \Big)} \frac{\Gamma^2 \Big( m_{\mathrm{f}} + \displaystyle{\frac{1}{2}} \Big)}{\Gamma^2 \Big( m_{\mathrm{f}} \Big)},
\end{align}
respectively. By applying \textbf{Lemma \ref{moment_matching}}, we can approximate the distribution of $\xi_{\mathrm{r}}$ as a Gamma distribution, i.e., $\xi_{\mathrm{r}} \sim \mathrm{Gamma} \big( \alpha_{\xi_{\mathrm{r}}},\beta_{\xi_{\mathrm{r}}} \big)$, where
\begin{align}
  \label{xi_r_shape}
  \alpha_{\xi_{\mathrm{r}}} &= \frac{N \Gamma^2 \Big( m_{\mathrm{g}} + \displaystyle{\frac{1}{2}} \Big) \Gamma^2 \Big( m_{\mathrm{f}} + \displaystyle{\frac{1}{2}} \Big)}{m_{\mathrm{g}} m_{\mathrm{f}} \Gamma^2 \Big( m_{\mathrm{g}} \Big) \Gamma^2 \Big( m_{\mathrm{f}} \Big) - \Gamma^2 \Big( m_{\mathrm{g}} + \displaystyle{\frac{1}{2}} \Big) \Gamma^2 \Big( m_{\mathrm{f}} + \displaystyle{\frac{1}{2}} \Big)} {\text{ and}} \notag\\
  \beta_{\xi_{\mathrm{r}}} &= \frac{m_{\mathrm{g}} m_{\mathrm{f}} \Gamma^2 \Big( m_{\mathrm{g}} \Big) \Gamma^2 \Big( m_{\mathrm{f}} \Big) - \Gamma^2 \Big( m_{\mathrm{g}} + \displaystyle{\frac{1}{2}} \Big) \Gamma^2 \Big( m_{\mathrm{f}} + \displaystyle{\frac{1}{2}} \Big)}{\sqrt{m_{\mathrm{g}} m_{\mathrm{f}}} \Gamma \Big( m_{\mathrm{g}} + \displaystyle{\frac{1}{2}} \Big) \Gamma \Big( m_{\mathrm{f}} + \displaystyle{\frac{1}{2}} \Big) \Gamma \Big( m_{\mathrm{g}} \Big) \Gamma \Big( m_{\mathrm{f}} \Big)},
\end{align}
are the corresponding shape and scaling parameters, respectively, obtained by substituting \eqref{xi_r_1st} and \eqref{xi_r_2rd} into \eqref{shape_scaling}. Then, we yield the $k$-th moment of $\xi_{\mathrm{r}}$, which is given by \cite{Tahir}
\begin{align}
  \label{xi_r_k}
  \mathbb{E} \Big\{ \xi_{\mathrm{r}}^k \Big\} &= \beta_{\xi_{\mathrm{r}}}^k \frac{ \Gamma \Big( \alpha_{\xi_{\mathrm{r}}} + k \Big)}{\Gamma \Big( \alpha_{\xi_{\mathrm{r}}} \Big)} \notag\\
  & = \beta_{\xi_{\mathrm{r}}}^k \alpha_{\xi_{\mathrm{r}}} \big( \alpha_{\xi_{\mathrm{r}}} + 1 \big) \cdots \big( \alpha_{\xi_{\mathrm{r}}} + k - 1 \big).
\end{align}

On the other hand, we note that term $\xi_{\mathrm{d}}$ follows the Nakagami-$m$ distribution, i.e., $\xi_{\mathrm{d}} \sim \mathrm{Nakagami} \big( m_{\mathrm{h}},1 \big)$. As such, the $k$-th moment of term $\xi_{\mathrm{d}}$ can be easily expressed as
\begin{align}
  \label{xi_d_k}
  \mathbb{E} \Big\{ \xi_{\mathrm{d}}^k \Big\} = \Big( \frac{1}{\sqrt{m_{\mathrm{h}}}} \Big)^k \frac{ \Gamma \Big( m_{\mathrm{h}} + \displaystyle{\frac{k}{2}} \Big)}{\Gamma \Big( m_{\mathrm{h}} \Big)}.
\end{align}

Thanks to the above results in \eqref{xi_r_k} and \eqref{xi_d_k}, we have
\begin{align}
  \label{xi_1st}
  {\mathbb{E}} \Big\{ \xi \Big\} &= {\mathbb{E}} \Big\{ \big( {\sqrt{\ell_{\mathrm{h}}} \xi_{\mathrm{d}} + \sqrt{\ell_{\mathrm{f}} \ell_{\mathrm{g}}} \xi_{\mathrm{r}}} \big)^2 \Big\} \notag\\
  &= \ell_{\mathrm{h}} {\mathbb{E}} \Big\{ \xi_{\mathrm{d}}^2 \Big\} + \ell_{\mathrm{f}} \ell_{\mathrm{g}} {\mathbb{E}} \Big\{ \xi_{\mathrm{r}}^2 \Big\} + 2 \sqrt{\ell_{\mathrm{h}} \ell_{\mathrm{f}} \ell_{\mathrm{g}}} {\mathbb{E}} \Big\{ \xi_{\mathrm{d}} \Big\} {\mathbb{E}} \Big\{ \xi_{\mathrm{r}} \Big\}, \\
  \label{xi_2rd}
  {\mathbb{E}} \Big\{ \xi^2 \Big\} &= {\mathbb{E}} \Big\{ \big( {\sqrt{\ell_{\mathrm{h}}} \xi_{\mathrm{d}} + \sqrt{\ell_{\mathrm{f}} \ell_{\mathrm{g}}} \xi_{\mathrm{r}}} \big)^4 \Big\} \notag\\
  &= \ell_{\mathrm{h}}^2 {\mathbb{E}} \Big\{ \xi_{\mathrm{d}}^4 \Big\} + \ell_{\mathrm{f}}^2 \ell_{\mathrm{g}}^2 {\mathbb{E}} \Big\{ \xi_{\mathrm{r}}^4 \Big\} + 6 \ell_{\mathrm{h}} \ell_{\mathrm{f}} \ell_{\mathrm{g}} {\mathbb{E}} \Big\{ \xi_{\mathrm{d}}^2 \Big\} {\mathbb{E}} \Big\{ \xi_{\mathrm{r}}^2 \Big\} \notag\\
  & \ \ \ + 4 \ell_{\mathrm{h}} \sqrt{\ell_{\mathrm{h}} \ell_{\mathrm{f}} \ell_{\mathrm{g}}} {\mathbb{E}} \Big\{ \xi_{\mathrm{d}}^3 \Big\} {\mathbb{E}} \Big\{ \xi_{\mathrm{r}} \Big\} \notag\\
  & \ \ \ + 4 \ell_{\mathrm{f}} \ell_{\mathrm{g}} \sqrt{\ell_{\mathrm{h}} \ell_{\mathrm{f}} \ell_{\mathrm{g}}} {\mathbb{E}} \Big\{ \xi_{\mathrm{d}} \Big\} {\mathbb{E}} \Big\{ \xi_{\mathrm{r}}^3 \Big\}.
\end{align}
By substituting \eqref{xi_1st} into \eqref{minimum_bit_energy_lowpower} and substituting \eqref{xi_1st}--\eqref{xi_2rd} into \eqref{wideband_slope_lowpower}, respectively, we obtain the desired results, which completes the proof of the corollary.

\subsection{Proof of Corollary \ref{discrete_phase_shifts}}

To further describe the minimum bit energy in \eqref{minimum_bit_energy_large_number_elements} and the wideband slope in \eqref{wideband_slope_large_number_elements} for the considered system in the low-power regime, it is necessary to derive the distribution of $\bar \xi = \big( \sqrt{\ell_{\mathrm{h}}} \xi_{\mathrm{d}} + \sqrt{\ell_{\mathrm{f}} \ell_{\mathrm{g}}} \bar \xi_{\mathrm{r}} \big)^2$, especially to derive the one of $\bar \xi_{\mathrm{r}} = \sum\nolimits_{n\in {\mathcal{N}}} \big| f_n \big| \big| g_n \big| {\mathrm{exp}} \big({\jmath{\hat \theta_n}} \big)$. As can be observed, $|f_n|$, $|g_n|$, and ${\mathrm{exp}} \big({\jmath{\hat \theta_n}} \big)$ are independent with each other as well as $\mathbb{E} \big\{ {\mathrm{exp}} \big({\jmath{\hat \theta_n}} \big) \big\} = \mathbb{E} \big\{ {\mathrm{exp}} \big({- \jmath{\hat \theta_n}} \big) \big\} = 2^b/ \pi \sin \big( \pi / 2^b \big)$ \cite{Qingqing2}. As such, for a given sufficiently large $N$, the first and second moments of term $\bar \xi_{\mathrm{r}}$ can be first given by
\begin{align}
  \label{bar_xi_r_1st}
  \mathbb{E} \Big\{ \bar \xi_{\mathrm{r}} \Big\} &= \frac{N}{\sqrt{m_{\mathrm{g}} m_{\mathrm{f}}}} \frac{ \Gamma \Big( m_{\mathrm{g}} + \displaystyle{\frac{1}{2}} \Big)}{\Gamma \Big( m_{\mathrm{g}} \Big)} \frac{ \Gamma \Big( m_{\mathrm{f}} + \displaystyle{\frac{1}{2}} \Big)}{\Gamma \Big( m_{\mathrm{f}} \Big)} \frac{2^b}{\pi} \sin \Big( \frac{\pi}{2^b} \Big) {\text{ and}} \notag\\
  \mathbb{E} \Big\{ \bar \xi_{\mathrm{r}}^2 \Big\} &= N + \frac{N (N - 1)}{m_{\mathrm{g}} m_{\mathrm{f}}} \frac{\Gamma^2 \Big( m_{\mathrm{g}} + \displaystyle{\frac{1}{2}} \Big)}{\Gamma^2 \Big( m_{\mathrm{g}} \Big)} \frac{\Gamma^2 \Big( m_{\mathrm{f}} + \displaystyle{\frac{1}{2}} \Big)}{\Gamma^2 \Big( m_{\mathrm{f}} \Big)} \notag\\
  & \ \ \ \times \Big( \frac{2^b}{\pi} \sin \Big( \frac{\pi}{2^b} \Big) \Big)^2,
\end{align}
respectively. By applying \textbf{Lemma \ref{moment_matching}}, we can approximate the distribution of $\bar \xi_{\mathrm{r}}$ as a Gamma distribution, i.e., $\bar \xi_{\mathrm{r}} \sim \mathrm{Gamma} \big( \alpha_{\bar \xi_{\mathrm{r}}},\beta_{\bar \xi_{\mathrm{r}}} \big)$, where $\alpha_{\bar \xi_{\mathrm{r}}}$ and $\beta_{\bar \xi_{\mathrm{r}}}$ are the corresponding shape and scaling parameters, respectively, obtained by substituting \eqref{bar_xi_r_1st} into \eqref{shape_scaling}. Then, we yield the $k$-th moment of $\bar \xi_{\mathrm{r}}$, which is expressed as \cite{Tahir}
\begin{align}
  \label{bar_xi_r_k}
  \mathbb{E} \Big\{ \bar \xi_{\mathrm{r}}^k \Big\} = \beta_{\bar \xi_{\mathrm{r}}}^k \alpha_{\bar \xi_{\mathrm{r}}} \big( \alpha_{\bar \xi_{\mathrm{r}}} + 1 \big) \cdots \big( \alpha_{\bar \xi_{\mathrm{r}}} + k - 1 \big).
\end{align}
By substituting \eqref{bar_xi_r_k} into \eqref{minimum_bit_energy_large_number_elements} and \eqref{wideband_slope_large_number_elements}, respectively, we have the desired results, which completes the proof of the corollary.

\subsection{Proof of Theorem \ref{wideband}}

The minimum bit energy for the considered system under statistical QoS constraints in the wideband regime is given by
\begin{align}
  \label{min_bit_energy_bandwidth}
  {\frac{E_{\mathrm{b}}}{N_0}}_{\min} &= \lim\limits_{B_{\mathrm{c}} \rightarrow \infty} \frac{1}{C_{\mathrm{E}} \big( B_{\mathrm{c}} \big) } \frac{P}{N_0 N_{\mathrm{c}} B_{\mathrm{c}} } \notag\\
  & = \lim\limits_{B_{\mathrm{c}} \rightarrow \infty} \frac{- \displaystyle{\frac{\mu T P}{N_0 N_{\mathrm{c}} {\mathrm{ln}}2 }}}{{\mathrm{ln}} \Big( \mathbb{E} \Big\{ {\mathrm{exp}} \Big( {- \mu T B_{\mathrm{c}} {\mathrm{log}}_2 \Big( 1 + \displaystyle{\frac{P }{N_0 N_{\mathrm{c}}  B_{\mathrm{c}}}} \xi } \Big) \Big) \Big\} \Big)} \notag\\
  & \overset{(f)} = - \frac{\displaystyle{\frac{\mu T P}{N_0 N_{\mathrm{c}} {\mathrm{ln}}2 }}}{{\mathrm{ln}} \Big( \mathbb{E} \Big\{ {\mathrm{exp}} \Big( {- \displaystyle{\frac{\mu T P}{N_0 N_{\mathrm{c}} {\mathrm{ln}}2 } } \xi} \Big) \Big\} \Big)},
\end{align}
where ($f$) can be obtained by the fact that $\lim \limits_{x \rightarrow \infty} x {\mathrm{ln}} \big( 1 + a/x \big) = a$ for any constant $a > 0$. Also, the first and second derivatives of the normalized effective capacity with respect to $B_{\mathrm{c}}$ in \eqref{normalized_C_E_wideband} are given by \eqref{dot_C_E_B_c} and \eqref{ddot_C_E_B_c} at the top of next page,
\begin{figure*}[ht]
\begin{align}
  \label{dot_C_E_B_c}
  \dot{C}_{\mathrm{E}} \big( B_{\mathrm{c}} \big) &= - \frac{1}{\mu T} {\mathrm{ln}} \bigg( \mathbb{E} \bigg\{ { {\mathrm{exp}} \Big( {- \mu T B_{\mathrm{c}} {\mathrm{log}}_2 \Big( 1 + \displaystyle{\frac{P }{N_0 N_{\mathrm{c}} B_{\mathrm{c}}}} \xi \Big) } \Big) } \bigg\} \bigg) - \mathbb{E} \bigg\{ {\mathrm{exp}} \Big( {- \mu T B_{\mathrm{c}} {\mathrm{log}}_2 \Big( 1 + \displaystyle{\frac{P }{N_0 N_{\mathrm{c}} B_{\mathrm{c}}}} \xi \Big) } \Big) \notag\\
  & \ \ \ \times \bigg( B_{\mathrm{c}} {\mathrm{log}}_2 \Big( 1 + \displaystyle{\frac{P }{N_0 N_{\mathrm{c}} B_{\mathrm{c}}}} \xi \Big) - \displaystyle{\frac{N_{\mathrm{c}} B_{\mathrm{c}} P \xi}{\big(N_0 N_{\mathrm{c}} B_{\mathrm{c}} + P \xi \big) {\mathrm{ln}}2}} \bigg) \bigg\} \bigg( {\mathbb{E} \Big\{ { {\mathrm{exp}} \Big( {- \mu T B_{\mathrm{c}} {\mathrm{log}}_2 \Big( 1 + \displaystyle{\frac{P }{N_0 N_{\mathrm{c}} B_{\mathrm{c}}}} \xi \Big) } \Big) } \Big\}} \bigg)^{-1}. \\
  \label{ddot_C_E_B_c}
  \ddot{C}_{\mathrm{E}} \big( B_{\mathrm{c}} \big) & = \mu T B_{\mathrm{c}} \bigg( \mathbb{E} \bigg\{ {\mathrm{exp}} \Big( {- \mu T B_{\mathrm{c}} {\mathrm{log}}_2 \Big( 1 + \displaystyle{\frac{P }{N_0 N_{\mathrm{c}} B_{\mathrm{c}}}} \xi \Big) } \Big) \bigg( B_{\mathrm{c}} {\mathrm{log}}_2 \Big( 1 + \displaystyle{\frac{P }{N_0 N_{\mathrm{c}} B_{\mathrm{c}}}} \xi \Big) - \displaystyle{\frac{N_{\mathrm{c}} B_{\mathrm{c}} P \xi}{\big(N_0 N_{\mathrm{c}} B_{\mathrm{c}} + P \xi \big) {\mathrm{ln}}2}} \bigg) \bigg\} \bigg)^2  \notag\\
  & \ \ \ \times \bigg( \mathbb{E} \Big\{ { {\mathrm{exp}} \Big( {- \mu T B_{\mathrm{c}} {\mathrm{log}}_2 \Big( 1 + \displaystyle{\frac{P }{N_0 N_{\mathrm{c}} B_{\mathrm{c}}}} \xi \Big) } \Big) } \Big\} \bigg)^{-2} - \bigg( {\mathbb{E} \Big\{ { {\mathrm{exp}} \Big( {- \mu T B_{\mathrm{c}} {\mathrm{log}}_2 \Big( 1 + \displaystyle{\frac{P }{N_0 N_{\mathrm{c}} B_{\mathrm{c}}}} \xi \Big) } \Big) } \Big\}} \bigg)^{-1} \notag\\
  & \ \ \ \times \mathbb{E} \bigg\{ {\mathrm{exp}} \Big( {- \mu T B_{\mathrm{c}} {\mathrm{log}}_2 \Big( 1 + \displaystyle{\frac{P }{N_0 N_{\mathrm{c}} B_{\mathrm{c}}}} \xi \Big) } \Big) \bigg( \mu T B_{\mathrm{c}} \Big( {\mathrm{log}}_2 \Big( 1 + \displaystyle{\frac{P }{N_0 N_{\mathrm{c}} B_{\mathrm{c}}}} \xi \Big) - \displaystyle{\frac{B P \xi}{\big(N_0 N_{\mathrm{c}} B_{\mathrm{c}} + P \xi \big) {\mathrm{ln}}2}} \Big)^2\notag\\
  & \ \ \ + {\mathrm{ln}}2 \Big( \displaystyle{\frac{N_{\mathrm{c}} B_{\mathrm{c}} P \xi}{\big(N_0 N_{\mathrm{c}} B_{\mathrm{c}} + P \xi \big) {\mathrm{ln}}2}} \Big)^2 \bigg) \bigg\}.
\end{align}\hrulefill\vspace*{-2mm}
\end{figure*}
respectively. Note that the first and second derivatives of the normalized effective capacity at $B_{\mathrm{c}} \rightarrow \infty$ are given by
\begin{align}
  \label{dot_C_E_B}
  \lim \limits_{B_{\mathrm{c}} \rightarrow \infty} \dot{C}_{\mathrm{E}} \big( B_{\mathrm{c}} \big) &= - \frac{1}{\mu T} {\mathrm{ln}} \bigg( \mathbb{E} \bigg\{ {\mathrm{exp}} \Big( {- \displaystyle{\frac{\mu T P}{N_0 N_{\mathrm{c}} {\mathrm{ln}}2 } } \xi} \Big) \bigg\} \bigg) {\text{ and}} \notag\\
  \lim \limits_{B_{\mathrm{c}} \rightarrow \infty} \ddot{C}_{\mathrm{E}} \big( B_{\mathrm{c}} \big) &= - \frac{P^2}{N_0^2 {\mathrm{ln}}2} \frac{\mathbb{E} \Big\{ \xi^2 {\mathrm{exp}} \Big( {- \displaystyle{\frac{\mu T P}{N_0 N_{\mathrm{c}} {\mathrm{ln}}2 } } \xi} \Big) \Big\} }{\mathbb{E} \Big\{ { {\mathrm{exp}} \Big( {- \displaystyle{\frac{\mu T P}{N_0 N_{\mathrm{c}} {\mathrm{ln}}2 } } \xi} \Big) } \Big\}},
\end{align}
respectively. By substituting \eqref{dot_C_E_B} into \eqref{wideband_slope}, we can acquire the desired results, which completes the proof of the theorem.

\subsection{Proof of Corollary \ref{wideband_large_number_elements}}

By applying \textbf{Lemma \ref{moment_matching}}, we first approximate the distribution of $\xi$ as a Gamma distribution, i.e., $\xi \sim \mathrm{Gamma} \big( \alpha_{\xi},\beta_{\xi} \big)$, where $\alpha_{\xi}$ and $\beta_{\xi}$ are the corresponding shape and scaling parameters, respectively, obtained by substituting \eqref{xi_1st} and \eqref{xi_2rd} into \eqref{shape_scaling}. Then, we express the probability density function (PDF) of $\xi$ as
\begin{align}
  \label{pdf_xi}
  f \big( \xi \big) = \frac{1}{\beta_{\xi}^{\alpha_{\xi}} \Gamma \big( \alpha_{\xi} \big)} {\xi}^{\alpha_{\xi} - 1} \exp \Big( {- \frac{\xi}{\beta_{\xi}}} \Big).
\end{align}
According to \eqref{pdf_xi}, we have \eqref{E_exp_xi} and \eqref{E_xi_exp_xi} at the top of next page,
\begin{figure*}[ht]
\begin{align}
  \label{E_exp_xi}
  \mathbb{E} \Big\{ {\mathrm{exp}} \Big( {- \displaystyle{\frac{\mu T P}{N_0 N_{\mathrm{c}} {\mathrm{ln}}2 } } \xi} \Big) \Big\} & = \int \nolimits_{0}^{\infty} {\mathrm{exp}} \Big( {- \displaystyle{\frac{\mu T P}{N_0 N_{\mathrm{c}} {\mathrm{ln}}2 } } \xi} \Big) f \big( \xi \big) \ {\mathrm{d}} \xi = \frac{1}{\beta_{\xi}^{\alpha_{\xi}} \Gamma \big( \alpha_{\xi} \big)} \int \nolimits_{0}^{\infty} {\xi}^{\alpha_{\xi} - 1} {\mathrm{exp}} \Big( {- \Big( {\displaystyle{\frac{\mu T P}{N_0 N_{\mathrm{c}} {\mathrm{ln}}2}}} + \frac{1}{\beta_{\xi}} \Big) \xi } \Big) \ {\mathrm{d}} \xi \notag\\
  & \overset{(l)} = \frac{1}{\beta_{\xi}^{\alpha_{\xi}} \Gamma \big( \alpha_{\xi} \big)} \Big( {\displaystyle{\frac{\mu T P}{N_0 N_{\mathrm{c}} {\mathrm{ln}}2}}} + \frac{1}{\beta_{\xi}} \Big)^{- \alpha_{\xi}} \Gamma \big( \alpha_{\xi} \big) = \Big( {\displaystyle{\frac{\mu T P \beta_{\xi}}{N_0 N_{\mathrm{c}} {\mathrm{ln}}2}} + 1} \Big)^{- \alpha_{\xi}}. \\
  \label{E_xi_exp_xi}
  \mathbb{E} \Big\{ \xi^2 {\mathrm{exp}} \Big( {- \displaystyle{\frac{\mu T P}{N_0 N_{\mathrm{c}} {\mathrm{ln}}2 } } \xi} \Big) \Big\} & = \int \nolimits_{0}^{\infty} \xi^2 {\mathrm{exp}} \Big( {- \displaystyle{\frac{\mu T P}{N_0 N_{\mathrm{c}} {\mathrm{ln}}2 } } \xi} \Big) f \big( \xi \big) \ {\mathrm{d}} \xi = \frac{1}{\beta_{\xi}^{\alpha_{\xi}} \Gamma \big( \alpha_{\xi} \big)} \int \nolimits_{0}^{\infty} {\xi}^{(\alpha_{\xi} + 2) - 1} {\mathrm{exp}} \Big( {- \Big( {\displaystyle{\frac{\mu T P}{N_0 N_{\mathrm{c}} {\mathrm{ln}}2}}} + \frac{1}{\beta_{\xi}} \Big) \xi } \Big) \ {\mathrm{d}} \xi \notag\\
  & \overset{(m)} = \frac{1}{\beta_{\xi}^{\alpha_{\xi}} \Gamma \big( \alpha_{\xi} \big)} \Big( {\displaystyle{\frac{\mu T P}{N_0 N_{\mathrm{c}} {\mathrm{ln}}2}}} + \frac{1}{\beta_{\xi}} \Big)^{- (\alpha_{\xi} + 2)} \Gamma \big( \alpha_{\xi} + 2 \big) = \alpha_{\xi} \big( \alpha_{\xi} + 1 \big) \beta_{\xi}^2 \Big( {\displaystyle{\frac{\mu T P \beta_{\xi}}{N_0 N_{\mathrm{c}} {\mathrm{ln}}2}} + 1} \Big)^{- (\alpha_{\xi} + 2)}.
\end{align}\hrulefill\vspace*{-4mm}
\end{figure*}
where ($l$) and ($m$) can be obtained by the fact that $\int \nolimits_{0}^{\infty} x^{a - 1} {\mathrm{exp}} \big( {- b x } \big) {\mathrm{d}} x = b^{- a} \Gamma \big( a \big)$ for any constants $a > 0$ and $b > 0$ \cite[Eq. 3.381.4]{Gradshteyn}. By substituting \eqref{E_exp_xi} into \eqref{minimum_bit_energy_wideband} and substituting \eqref{E_exp_xi}--\eqref{E_xi_exp_xi} into \eqref{wideband_slope_wideband}, respectively, we acquire the desired results, which completes the proof of the corollary.

\end{appendices}

\bibliographystyle{IEEEtran}
\bibliography{references}
\end{document}